\documentclass{article}
\usepackage[margin=1in]{geometry}
\usepackage[ascii]{inputenc}
\usepackage[colorlinks=true,linkcolor=blue,urlcolor=blue,citecolor=blue,anchorcolor=green]{hyperref}
\usepackage{bm,amsmath,amssymb,xcolor,amsthm,graphicx,enumitem}
\usepackage{tikz-cd}
\usepackage[capitalize]{cleveref}
\newtheorem{thm}{Theorem}[section]\crefname{thm}{Theorem}{Theorems}
\newtheorem{lem}[thm]{Lemma}\crefname{lem}{Lemma}{Lemmas}
\crefname{prb}{Problem}{Problems}
\crefname{conj}{Conjecture}{Conjectures}
\newtheorem{rem}[thm]{Remark}\crefname{rem}{Remark}{Remarks}
\newtheorem{cor}[thm]{Corollary}\crefname{cor}{Corollary}{Corollaries}
\newtheorem{dfn}[thm]{Definition}\crefname{dfn}{Definition}{Definitions}
\crefname{prp}{Proposition}{Propositions}
\newtheorem{obs}[thm]{Observation}\crefname{obs}{Observation}{Observations}
\newtheorem{exa}[thm]{Example}\crefname{exa}{Example}{Examples}
\DeclareMathOperator{\tr}{tr}

\DeclareMathOperator{\Mat}{Mat}

\DeclareMathOperator{\diag}{diag}

\DeclareMathOperator{\ch}{ch}
\DeclareMathOperator{\GL}{GL}

\DeclareMathOperator{\Herm}{Herm}
\DeclareMathOperator{\PD}{PD}

\DeclareMathOperator{\poly}{poly}
\DeclareMathOperator{\rk}{rank}

\DeclareMathOperator{\capa}{cap}
\DeclareMathOperator{\cut}{cut}

\DeclareMathOperator{\vol}{vol}

\DeclareMathOperator{\bbeta}{Beta}
\newcommand{\CC}{\mathbb C}
\newcommand{\Ess}{\mathbb S}

\newcommand{\EE}{\mathbb E}
\newcommand{\RR}{\mathbb R}

\newcommand{\ZZ}{\mathbb Z}

\newcommand{\eps}{\varepsilon}

\renewcommand{\vec}{\bm}
\begin{document}

\title{Rigorous Guarantees for Tyler's M-estimator via Quantum Expansion}
\author{
Cole Franks\thanks{Department of Mathematics, Massachusetts Institute of Technology. Email: \texttt{franks@mit.edu}.}
\and
Ankur Moitra\thanks{Department of Mathematics, Massachusetts Institute of Technology. Email: \texttt{moitra@mit.edu}. This work was
supported in part by a Microsoft Trustworthy AI Grant, NSF CAREER Award CCF-1453261, NSF Large CCF-
1565235, a David and Lucile Packard Fellowship, an Alfred P. Sloan Fellowship and an ONR Young Investigator
Award.}
}
\date{}
\maketitle

\begin{abstract}
Estimating the shape of an elliptical distribution is a fundamental problem in statistics. One estimator for the shape matrix, Tyler's M-estimator, has been shown to have many appealing asymptotic properties. It performs well in numerical experiments and can be quickly computed in practice by a simple iterative procedure. Despite the many years the estimator has been studied in the statistics community, there was neither a non-asymptotic bound on the rate of the estimator nor a proof that the iterative procedure converges in polynomially many steps.

Here we observe a surprising connection between Tyler's M-estimator and operator scaling, which has been intensively studied in recent years in part because of its connections to the Brascamp-Lieb inequality in analysis. We use this connection, together with novel results on quantum expanders, to show that Tyler's M-estimator has the optimal rate up to factors logarithmic in the dimension, and that in the generative model the iterative procedure has a linear convergence rate even without regularization. 
\end{abstract}




\section{Introduction}

The covariance matrix $\Sigma$ of a joint random variable $X$ is a fundamental object in statistics which encodes useful information about the geometry of the distribution of $X$. Estimation of the covariance matrix is a central task in data analysis, and in many situations the \emph{sample covariance matrix} is a good estimator. However, heavy-tailed random variables need not have a covariance matrix at all, and even when the covariance matrix exists the sample covariance matrix need not converge at all to the true one.

Elliptical distributions \cite{kelker1970distribution, cambanis1981theory} are a well-studied class of random variables used to model heavy-tailed data \cite{gupta2013elliptically}. Though elliptical distributions need not have covariance matrices, they are characterized by parameter called the \emph{shape matrix} with a similar geometric interpretation. Tyler \cite{Tyler_1987} defined an estimator, known as \emph{Tyler's M-estimator}, for the shape matrix of an elliptical distribution and proposed an iterative procedure to compute it. 
Furthermore, he established many powerful and surprising statistical properties for it. First, it is strongly consistent and asymptotically normal. In fact, its asymptotic distribution is \emph{distribution-free} in the sense that it that it does not depend on which elliptical distribution is generating the data. Second, it is the most robust estimator for the shape of an elliptical distribution in the sense that it minimizes the maximum asymptotic variance. 

There are some nonasymptotic bounds on the sample complexity \cite{soloveychik2014performance}, but these bounds are only comparable to the Gaussian case when error is measured in Frobenius norm. Moreover, the error has additional factors depending on the condition number of the shape matrix. From the computational standpoint, there are provably efficient algorithms for \emph{regularized} versions of the estimator \cite{Goes17}, but the regularized versions do not inherit the appealing statistical properties of Tyler's M-estimator. See the survey \cite{WZ15} for more information on Tyler's M-estimator and its regularized variants. Despite the significant attention the problem has received, the following are absent from the existing literature:
\begin{enumerate}
\item Necessary and sufficient conditions for the existence and uniqueness of Tyler's M-estimator.
\item\label{it:samples} A non-asymptotic upper bound on the sample complexity for estimation of the shape matrix in spectral norm comparable to the Gaussian case.
\item\label{it:linear} A rigorous proof that Tyler's iterative procedure converges at a linear rate. 
\end{enumerate}

Here we close all these gaps simultaneously by making a new connection, which surprisingly has gone unnoticed for decades, between Tyler's M-estimator and {\em operator scaling} \cite{gurvits2003classical}. We describe the setup for operator scaling in \cref{subsec:opscal}, but the name roughly refers to a group of problems generalizing Sinkhorn's classic ``matrix scaling" problem in which one seeks to obtain a doubly stochastic matrix by rescaling the rows and columns of a given nonnegative matrix. Tyler's M-estimator arises in operator scaling because it is precisely the ``rescaling" for a certain operator constructed from the samples.

Though the name is fairly new, operator scaling was studied far earlier in the context of \emph{geometric invariant theory} in algebra. As we will discuss in \cref{subsec:statement}, necessary and sufficient conditions for the existence of Tyler's M-estimator follow from the work \cite{king1994moduli} in this field. Later, the authors of \cite{gurvits2002deterministic, forster2002linear, barthe1998reverse} independently studied \cref{eq:tyler} for applications in convex geometry, communication complexity, and real analysis, respectively. It was shown in \cite{AM13}, and rather implicitly in \cite{gurvits2002deterministic}, that there is an algorithm to solve \cref{eq:tyler} up to error $\eps$ in time polynomial in $\log(1/\eps)$ and the bit-length of the samples, though both are very slow due to their use of the ellipsoid algorithm. Next, an iterative procedure for operator scaling was proposed in \cite{gurvits2003classical}, thus implicitly showing that Tyler's iterative procedure converges in time polynomial in $1/\eps$ and the bit-length of the samples. \cite{garg2016deterministic} proved the same guarantees in a significantly more general setting, and used them to obtain new upper bounds in algebraic complexity and then in \cite{garg2017algorithmic} to compute the Brascamp-Lieb constant in analysis.

Clearly a great deal of information about Tyler's M-estimator can be gleaned from the existing operator scaling literature, but there remain a few hurdles to \cref{it:samples,it:linear}. Firstly, it is unclear how well Tyler's estimator performs in a statistical sense in terms of proving finite sample guarantees. Secondly, results implying \emph{linear} ($\log \eps^{-1}$) convergence of iterative procedures like that of Tyler are rare, and the existing results are not explicit enough to produce polynomial time algorithms \cite{soules1991rate,knight2008sinkhorn} in the sense that they do not specify what norm they get convergence in. We clear these remaining hurdles through a somewhat surprising and subtle application of \emph{quantum expansion}, a tool from quantum information theory introduced recently to operator scaling in \cite{KLR19}. Moreover, we significantly sharpen the bounds in \cite{KLR19} in order to obtain an inverse exponential failure probability and an optimal rate of convergence up to logarithmic factors.

We now state our main results. Our first theorem plays the role in shape estimation analogous to the role of the matrix Chernoff theorem in covariance estimation, in the sense that it shows how well an ``empirical" version $\widehat{\Sigma}$ of the shape approximates the ``population" shape $\Sigma$. 
\begin{thm}[Sample complexity in spectral norm]\label{thm:elliptical}
For $n \geq C p \log^2 p /\eps^2$ samples from an elliptical distribution of shape $\Sigma$, Tyler's M-estimator $\widehat{\Sigma}$ satisfies 
$$\|I_p -  \Sigma^{1/2} \widehat{\Sigma}^{-1} \Sigma^{1/2}\|_{op} \leq \eps $$
with probability $1 - \exp( - \Omega( n \eps^2/\log^2 p))$ provided $\eps$ is a small enough constant. Here $C > 0$ is an absolute constant.
\end{thm}
This theorem provides guarantees for estimating the shape in the spectral norm, which was left as an open question in \cite{soloveychik2014performance}. Up to logarithmic factors, \cref{thm:elliptical} recovers the bound of \cite{soloveychik2014performance} on the Frobenius norm $\|\hat{\Sigma}^{-1} - \Sigma^{-1}\|_F$. They showed $\|\hat{\Sigma}^{-1} - \Sigma^{-1}\|_F \leq O( p / \gamma \sqrt{n} )$ with high probability, where $\gamma$ is a parameter depending on the spectrum of $\Sigma$.  We also prove guarantees for the Frobenius norm that, when the number of samples is large, are sharper than naively using \cref{thm:elliptical} to bound the Frobenius norm.
\begin{thm}[Sample complexity in Frobenius norm]\label{thm:elliptical-frob}
For $n \geq C p^2/\eps^2$ samples from an elliptical distribution of shape $\Sigma$, Tyler's M-estimator $\widehat{\Sigma}$ satisfies 
$$\|I_p -  \Sigma^{1/2} \widehat{\Sigma}^{-1} \Sigma^{1/2}\|_{F} \leq \eps $$
with probability $1 - \exp( - \Omega( n \eps^2/p^2))$ provided $\eps$ is a small enough constant. Here $C > 0$ is an absolute constant.
\end{thm}
Finally, we show that Tyler's iterative procedure converges quickly under mild conditions.
\begin{thm}[Iterative procedure]\label{thm:fast-sinkhorn-elliptical} For $n \geq C p \log^2 p$ samples from an elliptical distribution of shape $\Sigma$, Tyler's iterative procedure computes $\overline{\Sigma}$ satisfying 
\begin{gather}\|I_p - \widehat{\Sigma}^{1/2} \overline{\Sigma}^{-1} \widehat{\Sigma}^{1/2} \|_{F} \leq \eps\label{eq:sinkhorn-bound}\end{gather}
 in  $O(|\log \det \Sigma| + p +  \log (1/\eps))$
 iterations with probability $1 - \exp( - \Omega(n/p \log^2 p))$. 
\end{thm}
Moreover, in \cref{sec:finite-precision} we show that the above theorems hold even with access to finitely many bits of the samples. In light of the many recent results on operator scaling and its generalizations \cite{straszak2017computing,burgisser2017alternating,burgisser2018efficient, burgisser2019towards, KLR19}, we hope the connection we make here will shed light on other problems in statistical estimation.

\begin{rem}[Error metric]
One notes that we use the following measure of error:
\begin{gather} \|I_p -  \Sigma^{1/2} \widehat{\Sigma}^{-1} \Sigma^{1/2}\|_{op}.\label{eq:mahal}\end{gather}
We use this error because the \emph{Mahalanobis distance} between $\widehat{\Sigma},\Sigma$, given by \begin{gather}\|I_p -  \Sigma^{1/2} \widehat{\Sigma}^{-1} \Sigma^{1/2}\|_{F},\label{eq:real-mahal}
\end{gather} exceeds it by at most a factor $\sqrt{p}$. The Mahalanobis distance is a natural metric to study, because in the special case when $X$ is a Gaussian with covariance $\Sigma$, if the Mahalanobis distance is at most a small constant then it is on the order of the total variation distance between $X$ and the Gaussian with covariance $\widehat{\Sigma}$ \cite{barsov1987estimates}. Because $\Omega(p)$ samples are required to estimate the covariance of a Gaussian to constant spectral norm, and $\Omega(p^2)$ for constant Mahalanobis distance, \cref{thm:elliptical,thm:elliptical-frob} are tight up to logarithmic and constant factors, respectively. Moreover, the error metric \cref{eq:mahal} satisfies an approximate version of the triangle inequality (\cref{lem:triangle-ineq}), so approximating the estimator in this metric suffices to approximate the shape in the metric.
\end{rem}

\section{Tyler's M-estimator and operator scaling}

In this section we outline the connection between Tyler's M-estimator and operator scaling. As mentioned in the introduction, Tyler's M-estimator is precisely the ``rescaling" for a certain operator constructed from the samples. In some sense, the size of the rescaling governs the nearness of the estimator to the truth. Bounding the sizes of such scalings is a problem that arises naturally in operator scaling, and as shown by \cite{KLR19}, the size can be bounded by showing that the constructed operator is an \emph{approximate quantum expander}. In \cref{sec:quantum-expanders} we make this observation precise, and to prove \cref{thm:elliptical,thm:fast-sinkhorn-elliptical} we improve bounds on the probability that the operator is a quantum expander and show that the general form of Tyler's iterative procedure, known as \emph{Sinkhorn scaling}, converges in time $\poly(\log(1/\eps))$ on approximate quantum expanders.

\subsection{Elliptical distributions}\label{subsec:statement}
Consider $X$ drawn from a centered elliptical distribution $u \Sigma^{1/2} V$ on $\RR^p$, in which $u$ is random scalar independent of $V$, $\Sigma$ (known as the \emph{shape matrix}) is a fixed $p\times p$ positive-semidefinite matrix with $\tr \Sigma = 1$, and $V$ is a uniformly random element of $\Ess^{p - 1}$.

Our task is to find $\widehat{\Sigma}$ estimating $ \Sigma$. If $x_i, i\in [n]$ are drawn i.i.d from $X$, then Tyler's M-estimator $\widehat{\Sigma}$ for $ \Sigma$ is the defined to be the solution $\widehat{\Sigma}$ to the two equations
\begin{align} \frac{p}{n} \sum_{i =1}^n \frac{x_i x_i^\dagger}{x_i^\dagger \widehat{\Sigma}^{-1} x_i } &= \widehat{\Sigma}\label{eq:tyler}\\
\tr \widehat{\Sigma} &= p.\nonumber
\end{align} 
when the solution exists and is unique. 

It is known that if every $k$-dimensional subspace contains \emph{strictly less} than $k n/p$ vectors $x_i$, then Tyler's M-estimator exists and is unique, and also that if Tyler's M-estimator exists then every $k$-dimensional subspace contains \emph{at most} $k n/p$ vectors \cite{Tyler_1987}. Note the subtle difference between the two conditions. In fact, as shown by King in 1994 in a different language and without any apparent knowledge of Tyler's M estimator, the above sufficient condition is actually necessary. Below we include King's theorem stated in the language of Tyler's M estimator. We remark that \cref{it:tyler-uniqueness} in what follows does not require the full sophistication of King's theorem, so we provide a short proof in appendix \cref{sec:king} for completeness.
\begin{thm}[Existence of Tyler's M-estimator \cite{king1994moduli}] \label{thm:king}
Let $S = \{x_1, \dots, x_n\} \subset \RR^p$. 
\begin{enumerate}
\item \cref{eq:tyler} can be solved in $\PD(p)$ up to error $\eps$ for all $\eps > 0$ if and only if every $k$-dimensional subspace of $\RR^p$ contains at most $k n/p$ elements of $S$.
\item \cref{eq:tyler} has a solution in $\PD(p)$ if and only 
\begin{enumerate}
\item Every $k$-dimensional subspace of $\RR^p$ contains at most $k n/p$ elements of $S$, and
\item For every $k$-dimensional subspace $L$ of $\RR^p$ containing exactly $kn/p$ elements of $S$, there is a subspace $L'$ complementary to $L$ in $\RR^p$ containing the other $n - kn/p$ elements of $S$.
\end{enumerate}
\item\label{it:tyler-uniqueness} \cref{eq:tyler} has a unique solution in $\PD(p)$, or equivalently, Tyler's M-estimator exists, if and only if every proper $k$-dimensional subspace of $\RR^p$ contains strictly less than $kn/p$ elements of $S$. 
\end{enumerate}
\end{thm}

Moreover there is a simple iterative procedure for computing $\widehat{\Sigma}$ when it exists uniquely:

\begin{dfn}[Tyler's iterative procedure]\label{dfn:iterative} Set $\widehat{\Sigma}_0 = I_p$. For $t \in [T]$ set 
\begin{gather} \frac{p}{n} \sum_{i =1}^n \frac{x_i x_i^\dagger}{x_i^\dagger \widehat{\Sigma}_{t-1}^{-1} x_i } = \widehat{\Sigma}_{t}.\label{eq:tyler-fixedpt}
\end{gather}
Output $p \widehat{\Sigma}_{T}/\tr \widehat{\Sigma}_{T}$.

\end{dfn}
It is immediate that any fixed point of \cref{eq:tyler-fixedpt} is, up to a scalar multiple, Tyler's $M$ estimator $\widehat{\Sigma}$. Though \cite{Tyler_1987} also includes a normalization step, normalizing at the end of the procedure has the same effect and in any case the procedure in \cref{dfn:iterative} converges without normalization \cite{WZ15}.

\subsection{Operator scaling}\label{subsec:opscal}

The objects of study are \emph{completely positive maps}, maps $\Phi:\Mat(p) \to \Mat(n)$ between matrix spaces such that there exist $A_1, \dots, A_r \in \Mat(n \times p)$ such that $\Phi(Y) = \sum_{i = 1}^r A_i Y A_i^\dagger$. Completely positive maps arose in the study of $C^\star$ algebras, and play a role in quantum mechanics analogous to the role played by nonnegative matrices in classical probability. We will need some terminology. Let $\PD(p)$ denote the set of positive-definite $p\times p$ matrices.

\begin{dfn}[Completely positive maps]
Let $\Phi(X) = \sum_{i = 1}^r A_i X A_i^\dagger$ be a completely positive map.
\begin{itemize}
\item The \emph{size} of $\Phi$, denoted $s(\Phi)$, is given by $\tr \Phi(I_p).$
\item The \emph{dual} $\Phi^*:\Mat(n) \to \Mat(p)$ of $\Phi$ is its Hermitian adjoint and is given by $\Phi^*(X) = \sum_{i = 1}^r A_i X A_i^\dagger$.
\item Say $\Phi$ is \emph{$\eps$-doubly balanced} if 
$$\|\Phi(I_p)  - s(\Phi) I_n/n \|_{op} \leq s(\Phi)\eps/n\textrm{ and }\|\Phi^*(I_n)  - s(\Phi) I_p/p \|_{op} \leq s(\Phi)\eps/p,$$
and \emph{balanced} if it is $0$-doubly balanced. Up to scalar multiples, the balanced complely positive maps are exactly the \emph{unital quantum channels}.
\item If $L \in \GL(p), R \in \GL(n)$, let $ \Phi_{L,R}$, called a \emph{scaling} of $\Phi$ by $L,R$, denote the completely positive map given by
$$ \Phi_{L,R}:X \mapsto R \Phi(L^\dagger X L) R^\dagger.$$
Note that $(\Phi_{L,R})^* = \Phi_{R,L}.$
\end{itemize}
\end{dfn}
A central problem in operator scaling is the the existence, and computational efficiency of finding, doubly balanced scalings of a given operator. Analogously to Tyler's iterative procedure (\cref{dfn:iterative}), there is an iterative procedure to output $Z := L^\dagger L$. Note that if $\Phi_{L,R}$ is a balanced scaling of $\Phi$ then $R^\dagger R$ is equal to $\Phi(Z)^{-1}$, and $\phi_{\sqrt{L^\dagger L}, \sqrt{R^\dagger R}}$ is also balanced, so it is enough to look for scalings of the form $Z^{1/2}, \Phi(Z)^{-1/2}$. The following iterative procedure converges to such a $Z$ if it exists \cite{gurvits2003classical}.

\begin{dfn}[Sinkhorn's algorithm]\label{dfn:sinkhorn-alg}
Let $Z_0 \in \PD(p)$. 
For $t \in \ZZ_{\geq 1}$, set 
\begin{gather}\label{eq:fixedpt}Z_{t}^{-1} = \frac{p}{n} \Phi^*(\Phi(Z_{t-1})^{-1}).\end{gather}
We say $Z_i$ is the $t^{th}$ Sinkhorn iterate starting at $Z_0$. 
\end{dfn}
One can immediately check that $Z \in \PD(p)$ is a fixed point of Sinkhorn's algorithm if and only if $\Phi_{\sqrt{Z}, \Phi(Z)^{-1/2}}$ is balanced. 

 In the context of Tyler's M-estimator, the relevant example is a completely positive map $\Phi_{\vec x}$ constructed from a tuple $\vec x$ of vectors. This construction arises in the context of the Brascamp-Lieb inequalities in real analysis \cite{garg2017algorithmic}. We next see that finding a balanced scaling of the operator solves \cref{eq:tyler} and vice versa.
\begin{exa}[The map $\Phi_{\vec x}$.] For $\vec x \in (\CC^p)^n$, let 
$$\Phi_{\vec x}:X\mapsto  \diag(\langle x_i, X x_i \rangle: i \in [n]).$$
Then for $z \in \RR^n$, $\Phi_{\vec x}^*(\diag(z)) = \sum_{i = 1}^n z_i x_i x_i^\dagger.$ 
\end{exa}

One checks that $Z$ satisfies \cref{eq:fixedpt} for $\Phi = \Phi_{\vec x}$ if and only if $\widehat{\Sigma} = Z^{-1}$ satisfies \cref{eq:tyler-fixedpt}. Thus, Tyler's $M$ estimator is precisely the inverse of the first component of a scaling that balances $\Phi_{\vec x}$. Moreover, the output of $T$ steps of Tyler's iterative procedure is actually $pZ_T^{-1}/\tr Z_T^{-1}$ where $Z_T$ is the $T^{th}$ Sinkhorn iterate for $\Phi_{\vec x}$ starting at $I_p$.

As discussed earlier, the size of the scalings control the accuracy of the estimator, and we may control the size by showing that the operator is an approximate quantum expander. Roughly, a completely positive map is an approximate quantum expander if, as a linear map, its first singular vector is close to the identity matrix and its second singular value is strictly less than the first. 

\begin{dfn}[Appoximate quantum expanders] Let $\Phi$ be a completely positive map. Then $\Phi$ is a \emph{$(\eps, 1- \lambda)$-quantum expander} if it is $\eps$-doubly balanced and 
$$ \sup_{X \in \Herm(p): \tr X = 0} \frac{\|\Phi(X)\|_F}{\|X\|_F} \leq (1 - \lambda) \frac{s(\Phi)}{\sqrt{np}}.$$
A $(0, 1 - \lambda)$-quantum expander is called a $(1 -\lambda)$-\emph{quantum expander}.
\end{dfn}
We use a result of \cite{KLR19} relating expansion to the condition numbers of the scaling factors of an approximate quantum expander.

\begin{thm}[Follows from Theorem 1.7 of \cite{KLR19}]\label{thm:lapchi}
If $\Phi$ is an $(\eps, 1 -\lambda)$ quantum expander and $\eps \log p/\lambda^2$ is at most a small enough constant, then there are $L,R$ with $\det L = \det R = 1$ satisfying
$$\|I_p - L^\dagger L\|_{op}, \|I_n - R^\dagger R\|_{op} =  O(\eps \log p/\lambda)$$
such that $\Phi_{L,R}$ is doubly balanced.
\end{thm}

\cite{KLR19} actually studied a relaxed notion of approximate quantum expansion, called the \emph{spectral gap}. We discuss the relationship between spectral gap and approximate quantum expansion in \cref{sec:spectral-gap}, and show in particular that the hypotheses of Theorem 1.7 of \cite{KLR19} are equivalent to those of \cref{thm:lapchi}.

Next we define a function depending on a completely positive map that arises as a progress measure for the Sinkhorn scaling algorithm, and also for the specific case $\Phi = \Phi_{\vec x}$ is the quantity playing the role analogous to the log-likelihood for Tyler's $M$ estimator \cite{WZ15}.
\begin{dfn}[The capacity]\label{dfn:function}
Let $\Phi:\Mat(p) \to \Mat(n)$ be a completely positive map. Define the function $f^{\Phi}: \PD(p) \to \RR$ by 
$$ f^{\Phi}: \frac{p}{n} \log \det \Phi(Z) - \log \det Z.$$
The quantity $\inf_{Z \succ 0} f^{\Phi}$ is the logarithm of the \emph{capacity}, denoted $\capa(\Phi)$, defined in \cite{gurvits2003classical}.
\end{dfn}
The link between capacity and scalability is a central result in operator scaling.
\begin{thm}[\cite{gurvits2003classical}]\label{thm:gurvits} $\capa (\Phi) > 0$ if and only if $\Phi$ has an $\eps$-doubly balanced scaling for every $\eps > 0$. 
\end{thm}

This result suggests that to find Tyler's M-estimator to accuracy $\eps$ in time $\poly(p,n,\log(1/\eps)$, it suffices to find an $\eps$-minimizer of $f^{\Phi}$ in time $\poly(p,n,\log(1/\eps)$, and indeed this was also shown in \cite{gurvits2003classical}. We now record a few properties of the function $f^{\Phi}$.

\begin{itemize}
\item $f^{\Phi}$ is invariant under multiplication of the input by a positive scalar, i.e. 
\begin{gather}f^{\Phi}(Z) = f^{\Phi}(\alpha Z) \textrm{ for all }\alpha > 0, Z\in \PD(p).\label{eq:invar}\end{gather}
This shows that the progress measure doesn't change whether we use normalized or unnormalized Sinkhorn iterates. 
\item $f^{\Phi}$ has a property called \emph{geodesic convexity}: $f^{\Phi}(A^\dagger e^{t X} A)$ is convex in $t$ for all $X \in \Herm(p)$ and $A\in \GL(p)$. This property, which can be checked quite easily using the Cauchy-Binet formula and elementary calculus, was already observed for the case $\Phi = \Phi_{\vec x}$ (though not in this language) in \cite{WZ15}.
\end{itemize}
For a certain choice of metric on the manifold of positive definite matrices $\PD(p)$, this definition of geodesic convexity does match the usual one, in which a real valued function on a Riemmannian manifold is defined to be geodesically convex if it is convex along geodesics. We make use of the \emph{geodesic gradient} of a function $f:\PD(p) \to \RR$ at a point $Z$, given by 
\begin{gather}\nabla f(Z):= \nabla_X f(\sqrt{Z} e^{X} \sqrt{Z})|_{X = 0} \label{eq:geodesic-gradient}\end{gather}
Note that also $\nabla f^{\Phi}(Z) = \nabla f^{\Phi}(\alpha Z)$ for all $\alpha, Z\in \PD(p)$. For example, for $f = f^{\Phi}$ as defined in \cref{dfn:function} we have
 \begin{gather}\nabla f^{\Phi}(Z)= \frac{p}{n}\sqrt{Z}\Phi^*(\Phi(Z)^{-1}) \sqrt{Z} - I_p. \label{eq:operator-gradient}\end{gather}
 This identity gives intuition for \cref{thm:gurvits}: if $f^\Phi$ has a local minimum, then it is a critical point at which the geodesic gradient $\nabla f^\Phi $ vanishes. But if the right-hand side of \cref{eq:operator-gradient} is zero, then $Z$ is a fixed point of \cref{eq:fixedpt}, i.e. $\Phi_{Z^{1/2}, \Phi(Z)^{-1/2}}$ is doubly balanced. In particular, if $\Phi= \Phi_{\vec x}$ then we obtain 
 \begin{gather*}\nabla f^{\Phi_{\vec x}}(Z)= \frac{p}{n}\sum_{i = 1}^n \frac{\sqrt{Z} x_i x_i^\dagger \sqrt{Z}}{x_i^\dagger Z x_i} - I_p,\end{gather*}
which shows that if $f^{\Phi_{\vec v}}$ has a local minimum then \cref{eq:tyler} has a solution.

\subsection{Technical contributions}\label{subsec:results}
Having established the link between scalings and Tyler's M-estimator, we now state our result on quantum expansion, our result on iterative scaling of quantum expanders, and their corollary for Tyler's M-estimator.

We now state the results on expansion of $\Phi_{\vec v}$ which, in conjunction with \cref{thm:lapchi}, allows us to conclude \cref{thm:elliptical,thm:fast-sinkhorn-elliptical}. The following improves on \cite{KLR19} which showed the same but with failure probability $O(p/n^{3/4})$.

\begin{thm}\label{thm:random-expander} There are absolute constants $C, c, \lambda$ such that the following holds. Let $v_1, \dots, v_n$ be Haar random unit vectors from the sphere $\mathbb{S}^{p-1}$. Then $\Phi_{\vec v}$ is an $\left(\eps, 1 - \lambda\right)$ quantum expander with probability at least $1  - O(e^{-q(p,n,\eps)})$, where 
$$ q(p,n,\eps) = \min \{(c \sqrt n \eps - C \sqrt{p})^2, cn - C p \log p\}.$$
\end{thm}
We prove the theorem in \cref{sec:random}. We combine it with \cref{thm:lapchi} in \cref{sec:samples} to prove \cref{thm:elliptical}.

We next turn to algorithmic considerations. By the connection between optimization and scaling in \cref{subsec:opscal}, Tyler's M-estimator can be computed by finding an $\eps$-approximate minimizer to the convex function $f^{\Phi_{\vec x}^*}$. The works \cite{gurvits2002deterministic, AM13} use the ellipsoid method to accomplish this, showing Tyler's M-estimator can be computed up to accuracy $\eps$ in time $\poly(p,n,\log(1/\eps))$. Later in \cite{AGLOW18} it was shown that the optimization problem can be solved by second-order ``trust region" methods.

Though the algorithms of \cite{gurvits2002deterministic, AM13, AGLOW18} have polynomial time guarantees, both the ellipsoid method and trust regions tend to be slow. Tyler's iterative procedure, on the other hand, is very simple and fast in practice. We next discuss our results in this direction, which will be proven in \cref{subsec:sinkhorn-sc}. Firstly we show that Sinkhorn's algorithm converges in time $O(\log (1/\eps))$ for quantum expanders.
\begin{thm}\label{thm:expander-sinkhorn}
If $\Phi:\Mat(p) \to \Mat(n)$ is a $(1 - \lambda)$-quantum expander and $\|\nabla f^{\Phi}(Z_0)\|_F \leq c \lambda^2$, then the $T^{th}$ Sinkhorn iterate $Z_T$ starting at $Z_0$ satifies 
$$ \|\nabla f^{\Phi}(Z_T)\|_F = \exp( - O(\lambda T))$$
for some small enough constant $c$.
\end{thm}
 In \cref{subsec:sinkhorn-el} we straightforwardly combine this theorem with \cref{thm:random-expander} to show \cref{thm:fast-sinkhorn-elliptical} on the fast rate of convergence for Tyler's iterative procedure. 
 
 Finally, we prove a version of \cref{thm:lapchi} adapted for the Frobenius norm. In some cases, as for Tyler's M estimator, one wishes to obtain Frobenius bounds on $L^\dagger L$. One can simply combine the bound $\| \cdot \|_{F} \leq \sqrt{p}\| \cdot \|_{op}$ with the previous theorem, but our techniques yield a bound that is sharper by a logarithmic factor. Our result is incomparable to \cref{thm:lapchi} and appeals only to geodesic convexity rather than the more sophisticated methods of \cite{KLR19}. The theorem is proven in \cref{subsec:strong-convexity}.
\begin{thm}[Frobenius version of \cite{KLR19}]\label{thm:frob-improvement} Suppose $\Phi$ is a $(\eps, 1 - \lambda)$-quantum expander and $\eps \sqrt{p}/\lambda$ is at most a small enough constant. Then there are $L, R$ with $\det L = \det R = 1$ such that $\Phi_{L,R}$ is doubly balanced and 
$$ \|I_p - L^\dagger L \|_F, \|I_p - R^\dagger R \|_{op} = O(\eps \sqrt{p} /\lambda).$$ 
\end{thm}

\subsection{Organization of the paper}

In \cref{sec:samples} we prove \cref{thm:elliptical,thm:elliptical-frob} on the convergence of Tyler's M-estimator using quantum expansion. In \cref{sec:quantum-expanders} we show \cref{thm:frob-improvement} on condition number of scalings and \cref{thm:expander-sinkhorn} on the fast convergence of Sinkhorn's algorithm for quantum expanders, and use this to show that Tyler's iterative procedure converges quickly (\cref{thm:fast-sinkhorn-elliptical}). Both \cref{sec:samples} and \cref{sec:quantum-expanders} rely on \cref{thm:random-expander} on the quantum expansion of the operator $\Phi_{\vec v}$, the proof of which we delay until \cref{sec:random}. In the appendix we show relationships between approximate quantum expansion and the spectral gap defined in \cite{KLR19}, and extend \cref{thm:fast-sinkhorn-elliptical} and \cref{thm:elliptical} to the finite precision setting.

\section{Sample complexity bounds via quantum expansion}\label{sec:samples}


First, we observe that we only need to consider the accuracy of Tyler's M-estimator when the shape is the identity.
\begin{obs}\label{obs:identity} Suppose $x_i = u_i \Sigma^{1/2} v_i$. Then $Y$ solves \cref{eq:tyler} for $\vec x$ if and only if $Z = \Sigma^{-1/2}Y \Sigma^{-1/2}$ solves \cref{eq:tyler} for $\vec v$, and clearly
$$ \| I_p - Z^{-1}\|_{op} = \|I_p -  \Sigma^{1/2}Y^{-1} \Sigma^{1/2}\|_{op}.$$
\end{obs}
This means we need only show that Tyler's M-estimator is accurate for the elliptical distribution in which $u_i = 1$ and $\Sigma = I_p$. This does not violate our assumption that we do not know $\Sigma$, because will never need \emph{access} to the $v_i$'s. 

\begin{proof}[Proof of \cref{thm:elliptical} and \cref{thm:elliptical-frob}]
We first show \cref{thm:elliptical}. As $X$ is a centered elliptical distribution, $x_i = u_i \Sigma v_i$ for $v_1, \dots, v_n$ equal to Haar random unit vectors. Let $\lambda$ be a small constant as in \cref{thm:random-expander}. By \cref{thm:random-expander}, with probability at least $1  - O(e^{-q(p,n,\eps/\log p)}) = 1 - \exp( - \Omega( n \eps^2/\log^2 p))$ for $n \geq C p \log^2 p/\eps^2$ the operator $\Phi_{\vec v}$ is an $(\eps/\log p, 1 - \lambda)$-quantum expander. 

By \cref{thm:lapchi}, there is a solution $Z$ to \cref{eq:tyler} for $\vec v$ that satisfies $\|Z^{-1} - I_p\|_{op} = C\eps' \log p$. By \cref{obs:identity}, $Z = \Sigma^{-1/2} Y \Sigma^{-1/2}$ where $Y$ solves \cref{eq:tyler} for $\vec x$ and in particular 

\begin{gather}\|I_p -  \Sigma^{1/2} Y^{-1} \Sigma^{1/2}\|_{op} = O(\eps).\label{eq:pre-normalize}\end{gather}
Tyler's M-estimator is $\widehat{\Sigma} = p Y/\tr Y$. From \cref{eq:pre-normalize}, we have 
$ Y \in (1 + O(\eps)) \Sigma$
and as $\tr \Sigma = p$ we have $\tr Y = p(1 + O(\eps))$ so that 
\begin{gather*}\|I_p -  \Sigma^{1/2} \widehat{\Sigma}^{-1} \Sigma^{1/2}\|_{op} = O(\eps).\end{gather*}
By replacing $\eps$ by a suitable constant multiple we may assume the error bound is $\eps$ rather than $O(\eps)$. To show \cref{thm:elliptical-frob}, apply \cref{thm:random-expander} with $\eps/\sqrt{p}$ rather than $\eps/\log p$ and use \cref{thm:frob-improvement} in place of \cref{thm:lapchi}.
\end{proof}

\section{Fast Sinkhorn for quantum expanders}\label{sec:quantum-expanders}

As discussed in the introduction, quantum expansion is a natural way to bound sizes of scalings. Because scalings are optimizers of the objective function $f^\Phi$ of \cref{dfn:function} in \cref{subsec:opscal}, and bounding the size of optimizers is frequently accomplished through strong convexity, one is led to investigate the relationship between quantum expansion and strong convexity. 

In this section we prove such a relationship and use it to show that Tyler's iterative procedure converges linearly. Because Tyler's iterative procedure can be considered a descent method which makes progress proportional to the gradient of $f^\Phi$, it's straightforward to verify (see \cref{subsec:sinkhorn-sc}) that Sinkhorn scaling converges linearly provided $f^\Phi$ is strongly convex along the \emph{entire trajectory} of the procedure. We show that this is the case in \cref{subsec:strong-convexity} by showing that $f^{\Phi}$ is strongly convex in a suitably large \emph{sublevel set}. Finally in \cref{subsec:sinkhorn-el} we specialize to Tyler's iterative procedure.

\subsection{Strong convexity from quantum expansion}\label{subsec:strong-convexity}

By \cref{eq:invar} of $f^{\Phi}$, we can only hope for for strong geodesic convexity on the manifold of positive definite matrices with determinant $1$, which we call $\PD_1(p)$. 

\begin{rem}[Normalization]
The usual convention is $\tr \Sigma =  \tr \widehat{\Sigma} = p$, but to avoid more calculations like those at the end of the proof of \cref{thm:elliptical} it will more convenient to assume $\Sigma, \widehat{\Sigma}$ and the Sinkhorn iterates are in $\PD_1(p)$. For this reason if $A \in \PD(p)$ we let $A_1:=\det(A)^{-1/p} A  \in \PD_1(p)$. An elementary calculation (\cref{lem:trace-det}) shows that switching between the two normalizations for $\Sigma, \widehat{\Sigma}$ only incurs a constant factor error in \cref{eq:mahal} provided it is at most a small constant. If $Z_i$ are a sequence of Sinkhorn iterates, then the elements of the sequence $(Z_i)_1$ are called the \emph{normalized} Sinkhorn iterates. \end{rem}

We say a function $f:\PD_1(p) \to \RR$ is geodesically $\gamma$-strongly convex at $Z$ if for all Hermitian $X$ with $\tr X = 0$, we have
$$\partial_t^2 f\left(\sqrt{Z}e^{tX}\sqrt{Z}\right) \geq \gamma \|X\|_F^2.$$
A key feature of this definition is that for $f = f^{\Phi}$, different points $Z$ correspond to scalings of $\Phi$. More precisely, if $f^\Phi$ is $\gamma$ strongly convex at $Z$ if and only if $f^{\Phi_{Z, I_n}}$ is strongly convex at $I_p$.

\begin{lem}[Strong convexity from expansion] There are constants $C,c > 0$ such that for all $\eps < c$, if $\Phi:\Mat(p) \to \Mat(n)$ is a $(\eps, 1 -\lambda)$-quantum expander, then the function $X \mapsto  \log\det(\Phi(X))$ from $\PD_1(p)$ to $\RR$ is geodesically
$$\frac{n}{p} \left(2\lambda - \lambda^2 - C\eps \right)$$ strongly convex at the origin. Note that $2\lambda - \lambda^2\geq \lambda.$
\end{lem}

\begin{proof}Consider $X$ Hermitian with $\tr X = 0$. Observe that 
$$\partial_t^2 \log\det(\Phi(e^{tX}))|_{t = 0} =  \tr \Phi(I_p)^{-1} \Phi(X^2) - \tr \Phi(I_p)^{-1}\Phi(X)\Phi(I_p)^{-1}\Phi(X).$$
We lower bound the first term and upper bound the second. Because $\Phi$ is $\eps$-doubly balanced, $\Phi(I_p) \succeq (1 - \eps) (s(\Phi)/n) I_n,$ so using the monotonicity of the trace inner product in each argument under the Loewner ordering we have
\begin{align*}
\tr \Phi(I_p)^{-1}\Phi(X)\Phi(I_p)^{-1}\Phi(X)  \leq (s(\Phi)/n)^{-2}(1 - \eps)^{-1} \tr \Phi(X)^2.
\end{align*}
on the other hand, $\Phi^*(I_n) \succeq (1 - \eps)(s(\Phi)/p)I_p$ and $\Phi(I_p) \preceq (1 + \eps) (s(\Phi)/n) I_p$, so 
\begin{align*}\tr\Phi(I_p)^{-1}\Phi(X^2) &\geq (1 + \eps)^{-1}(s(\Phi)/n)^{-1}\tr \Phi(X^2)\\
&= (1 + \eps)^{-1}(s(\Phi)/n)^{-1}\tr \Phi(X^2)I_n\\
& = (1 + \eps)^{-1}(s(\Phi)/n)^{-1}\tr X^2 \Phi^*(I_n)\\
& \geq (1 + \eps)^{-1}\frac{n}{p}(1 - \eps)\tr X^2.
\end{align*} Combining these two bounds and using the quantum expansion of $\Phi$ and that $\eps < c$, we have
\begin{align*}\partial_t^2 \log\det(\Phi(e^{tX})) &\geq \frac{n}{p}\|X\|_F^2\left((1 + \eps)^{-1}(1 - \eps) - (1 - \lambda)^2(1 - \eps)^{-1}\right).\\
&\geq \frac{n}{p} \|X\|_F^2\left((2\lambda - \lambda^2)(1 - C\eps) - C\eps\right).\qedhere\end{align*}
\end{proof}
Because the function $Z \mapsto \log\det Z$ is geodesically linear, we have the following corollary.
\begin{cor}\label{cor:expansion-convexity}
If $\eps < c$ and $\Phi$ is a $(\eps, \lambda)$-quantum expander, then the function
$f^{\Phi}:\PD_1(p) \to \RR$ is geodesically $\lambda - C\eps$-strongly convex at the origin. Here $C, c> 0$ are absolute constants.
\end{cor}

While it is nice to know that $f^\Phi$ is geodesically strongly convex at $I_p$, to deduce bounds on the sizes of scalings we need to show that the function is strongly convex \emph{near} $I_p$. The next lemma is the first step in this direction.

\begin{lem}[Expansion under scaling]\label{lem:expansion-scaling} If $\Phi$ is an $(\eps < 1, 1 - \lambda)$-quantum expander and $\|L^\dagger L  - I_p\|_{op} \leq \delta$ and $\|R^\dagger R - I_n\|_{op} \leq \delta$ for $\delta < c$, then $\Phi_{L,R}$ is an $(\eps + C\delta, 1 - \lambda + C\delta)$-quantum expander. Here $C,c>0$ are absolute constants.

\end{lem}
\begin{proof}
We first show the balanced-ness condition. Because $(1 - \delta) I_p\leq L^\dagger L \leq (1 + \delta) I_p$ and $(1 - \delta) I_p \leq  R R^\dagger \leq (1 + \delta) I_n$, we have 
\begin{align}
(1 - \delta)^2 s(\Phi)\leq s(\Phi_{L,R}) \leq (1 + \delta)^2 s(\Phi).\label{eq:size}
\end{align}
Furthermore, 
\begin{align*}R \Phi(L^\dagger I_p L)R^\dagger 
 & \preceq (1 + \delta)^2 (1 + \eps)\frac{s(\Phi)}{n} I_n.\\
 & \preceq \frac{(1 + \delta)^2}{(1 - \delta)^2} (1 + \eps)\frac{s(\Phi_{L,R})}{n} I_n \end{align*}
 The other condition is similar. 
 
We next verify the expansion condition. We seek to bound the operator norm of $\Phi_{L,R}$ restricted to the traceless Hermitians. First note that the map $Y \to R^\dagger Y R$ has operator norm at most $(1 + \delta)$, so the desired bound will be the operator norm of $X \mapsto \Phi(L^\dagger X L)$ restricted to the traceless Hermitians multiplied by $(1 + \delta)$. Let $\pi(Z) = Z - (\tr Z)I_p/p$ denote the projection of $Z$ to its traceless part, and let $\Psi:X \mapsto \pi(L^\dagger X L)$. Note that $\Psi$ maps the traceless Hermitians to the traceless Hermitians, and $\Psi$ has operator norm at most $(1 +\delta)$. Then we may write 
 $$ \Phi(L^\dagger X L) = \Phi \circ \Psi (X) + \Phi(I_p) \tr (L^\dagger X L)/p.$$
By the triangle inequality, the operator norm of $X \mapsto \Phi(L^\dagger X L)$ restricted to the traceless Hermitian matrices is at most the sum of that of  $\Phi \circ \Psi $ and that of $\Gamma:X \mapsto \Phi(I_p) \tr (L^\dagger X L)/p$. The former is immediately seen to be at most $(1 - \lambda) (1 + \delta)s(\Phi)/\sqrt{np}$, and because $\tr X =0$ we have
\begin{align*}
\|\Gamma\|_{op} &\leq \frac{1}{p} \|\Phi(I_p)\|_F \sup_{\tr X = 0} \frac{\tr L^\dagger X L}{\|X\|_F}\\
& =  \frac{1}{p} \|\Phi(I_p)\|_F \sup_{\tr X = 0} \frac{\tr LL^\dagger X}{\|X\|_F}\\
& \leq \frac{1}{p} \|\Phi(I_p)\|_F \sup_{\tr X = 0} \frac{\tr (L^\dagger L - I_p) X}{\|X\|_F}\\
& \leq \frac{1}{p} \|\Phi(I_p)\|_F \sup_{\tr X = 0} \delta \sqrt p\\
&\leq (1 + \eps) s(\Phi) \delta/\sqrt{np},
\end{align*} 
where in the last condition we used that $\Phi$ is $\eps$-doubly balanced. Using \cref{eq:size} again, we have
$$
 \sup_{X \in \Herm(n), \tr X = 0} \frac{\|\Phi_{L,R}(X)\|_F^2}{\|X\|_F^2} \leq \frac{s(\Phi_{L,R})}{\sqrt{np} (1 - \delta)^2} \left((1 - \lambda)(1 + \delta) + (1 + \eps) \delta \right).
$$
Because $\delta \leq c$ and $\eps < 1$, we have the claim.
\end{proof}

\begin{cor}\label{cor:convex-ball}
There is an absolute constant $c$ such that the following holds. Suppose $\Phi$ is a $(\eps, 1 - \lambda)$-quantum expander and $\eps \leq c \lambda$. Then $f^\Phi$ is geodesically $\lambda/2$-strongly convex on the geodesic ball of radius $c\lambda$ about $I_p$.
\end{cor}

\begin{proof}
The strong convexity of $f^\Phi$ at a point $Z$ equals the strong convexity of $g:=f^{\Phi_{\sqrt{Z}, I_n}}$ at $I_p$.  For $\kappa > 0$ to be determined shortly, by \cref{lem:expansion-scaling}, if $\|Z - I_p\|_{op} \leq \kappa,$ then $\Phi_{\sqrt{Z}, I_n}$ is an $(\eps + C \kappa, 1 - \lambda + C \kappa$)-quantum expander. By \cref{cor:expansion-convexity}, for $C'$ a large enough absolute constant, $ f^{\Phi}$ is $\lambda - C'\eps - C'\kappa$-strongly convex at $Z$ provided $\|Z - I_p\|_{op} \leq \kappa,$ which is implied by $\|\log Z\|_F \leq c'\kappa$ for $\kappa,c'$ at most a small absolute constant.

If we take $\kappa = c \lambda$ for $c$ a small enough constant we conclude that $f^\Phi$ is $\lambda/2$ strongly convex on the geodesic ball of radius $c\lambda$ about $I_p$.
\end{proof}

We now include two lemmas with standard results from geodesically convex optimization. The first will imply a sharper bound for estimating the shape matrix in Frobenius norm. 
\begin{lem}\label{lem:frob-grad}
Suppose that $f:\PD_1(p) \to \RR$ is geodesically $\lambda$-strongly convex on the geodesic ball of radius $\kappa$ about $Z_0$ and that $\|\nabla f(Z_0) \|_F < \lambda \kappa $. 
Then there is a unique optimizer $Z^*$ of $f$ satisfying 
\begin{gather} \|\log Z_0^{-1/2} Z^* Z_0^{-1/2} \|_F \leq \frac{\|\nabla f(Z_0)\|_F}{\lambda}\label{eq:distance-bound} \end{gather}
and
\begin{gather}\quad f(Z^*) \geq f(Z_0) - \frac{1}{2\lambda} \|\nabla f(Z_0)\|_F^2.\label{eq:function-bound}\end{gather} 
\end{lem}
\begin{proof} By replacing the function with $X \mapsto f(\sqrt{Z_0} X \sqrt{Z_0})$, it is enough to prove the assertions in the case $Z_0 = I_p$. Let $B$ be the geodesic ball of radius $\kappa$ about $I_p$. By geodesic convexity on $B$, we have
\begin{align}\frac{d^2}{dt^2} f(e^{t H}) = \frac{d^2}{ds^2} f(e^{\frac{1}{2} t H} e^{s H} e^{\frac{1}{2} t H}) \geq \lambda \label{eq:second-deriv}.
\end{align}for all $t  \leq \kappa, \|H\|_F = 1$. Integrating \cref{eq:second-deriv} once, we have 
\begin{align}\frac{d}{dt} f(e^{t H}) \geq \frac{d}{dt} f(e^{t H})|_{t = 0} +  \lambda t  \geq -\|\nabla f(I_p)\|_F + \lambda t,\label{eq:first-deriv}
\end{align}
for $0 \leq t \leq \kappa, \|H\|_F= 1$. The last inequality is by Cauchy-Schwarz. Integrating again, we find that \begin{align}f(e^{t H}) \geq -t\|\nabla f(I_p)\|_F + \frac{\lambda}{2} t^2.\label{eq:zero-deriv}\end{align}
$0 \leq t \leq \kappa, \|H\|_F= 1$. Evaluating \cref{eq:zero-deriv} with $t = \kappa$ shows that on the boundary of $B$, the function value is at least $f(I_p)$. As a consequence, a local minimum $Z^*$ must occur in $B$, and by strong convexity it must be unique. Solving \cref{eq:first-deriv} for $t$ with $e^{tH} = Z^*$, and hence left-hand side equal to zero, shows $t \leq \|\nabla f(I_p)\|_F/\lambda \leq \kappa$, proving \cref{eq:distance-bound}. \cref{eq:zero-deriv} is minimum at this value of $t$; evaluating at the minimum implies \cref{eq:function-bound}. \end{proof}
We now apply the previous lemma to $f^{\Phi}$ to conclude \cref{thm:frob-improvement}.
\begin{proof}[Proof of \cref{thm:frob-improvement}] We first show, using the assumption that $\Phi$ is $\eps$-doubly balanced, that $\|\nabla f^{\Phi}(I_p)\|_F = O(\eps \sqrt{p} )$. Recall that $\nabla f^{\Phi}(I_p) = \frac{p}{n} \Phi^*(\Phi(I_p)^{-1}) - I_n.$ By $\eps$-balancedness and $\eps \leq c$, $\|(n\Phi(I_p)/s(\Phi))^{-1} - I_n \|_{op} = O(\eps)$, so $\|A\|_{op}:= \|\Phi(I_p)^{-1} - n I_p/s(\Phi) \|_{op} = O(n \eps/s(\Phi)).$ Now 
\begin{align*}
\|\frac{p}{n} \Phi^*(\Phi(I_p)^{-1}) - I_n\|_{op} &\leq \left\|\frac{p}{n} \Phi^*(n I_p/s(\Phi)) - I_n\right\|_{op} + \|A\|_{op}\\
&= O( \eps). 
\end{align*}
Thus $\|\nabla f^{\Phi}\|_F \leq \sqrt{p} \|\nabla f^{\Phi}\|_{op} = O(\eps \sqrt{p})$. By \cref{cor:convex-ball}, we may now apply \cref{lem:frob-grad} with $\kappa = c \lambda$ and $Z_0 = I_p$ to conclude that there is a unique optimizer $Z^*$ for $f^{\Phi}$ satisfying \begin{gather} \|\log Z^* \|_F = O (\sqrt{p} \eps/\lambda). \label{eq:op-distance-bound} \end{gather}
Because $\sqrt{p} \eps/\lambda \leq c$, $\| I - Z \|_F = O(\sqrt{p}\eps/\lambda).$ We may take $L = \sqrt{Z^*}$. It is a straightforward calculation to show that, by $\eps$-balancedness of $\Phi$, the scaling $R$ may be taken as some multiple of $\Phi(Z^*)^{-1/2}$. \end{proof}

The next lemma us useful for showing that the iterative algorithm converges quickly. In what follows we assume that the identity is the optimizer. The results are sharper in this case, and we may assume this is the case by using \cite{KLR19} to ``translate" the optimizer to the identity.
\begin{lem}\label{lem:grad} Suppose that $f:\PD_1(p) \to \RR$ is geodesically $\lambda$-strongly convex on the geodesic ball of radius $\kappa$ about $Z_0$, and that $\nabla f(Z_0) = 0$. If $\|\nabla f(Y)\|_F < \lambda\kappa/8$, then $Y$ is contained in a sublevel set of $f$ on which $f$ is geodesically $\lambda$-strongly convex. 
\end{lem}

\begin{proof}By replacing the function with $X \mapsto f(\sqrt{Z_0} X \sqrt{Z_0})$, it is enough to prove the lemma when $Z_0 = I_p$. Let $S$ be the largest level set contained in the geodesic ball $B$ of radius $\kappa$ about $I_p$. That is, if $h := \min_{X \in \partial B} f(X)$, then 
$$S = \{X:f(X) \leq h \}.$$ 
It suffices to show that if $Y$ is outside $S$, then $\|\nabla f(Y)\|_F \geq \kappa/8 \lambda$. Let $X$ be a point in $\partial B$, i.e. a point with $\|\log X\|_F = \kappa$, with $f(X) = h$. Let $H = \log X/\|\log X\|_F$.  By \cref{eq:zero-deriv}, we have $f(X) = h \geq \lambda \kappa^2/2$. Now let $Y \not \in S$. By continuity, there is a value $t_0$ between $0$ and $\kappa$ such that $f(e^{t K}) = h$ where $K = \log Y/\|\log Y\|_F$. By the mean value theorem, there is a value $t$ between $0$ and $\kappa$ such that $\frac{d}{dt} f(e^{tX}) \geq h/t_0 \geq \kappa/2$. By convexity and Cauchy-Schwarz, $\|\nabla f(Y)\|_F \geq \frac{d}{dt} f(e^{tX})|_{t = \kappa} \geq h/t_0 \geq \kappa/2$. \end{proof}

By \cref{cor:convex-ball}, we may apply \cref{lem:grad} with $\kappa = c \lambda$ and $Z_0 = I_p$.
\begin{cor}\label{lem:level-set}There is an absolute constant  $c > 0$ so that the following holds. Suppose $\Phi$ is a $(1 - \lambda)$-quantum expander. If 
$$\|\nabla f^{\Phi}(Z)\|_F \leq c\lambda^2,$$ then $f^\Phi: \PD_1(p) \to \RR$ is geodesically $\lambda/2$-strongly convex on a sublevel set $S$ of $f^\Phi$ containing $Z$.
\end{cor}
Applying \cref{lem:frob-grad} with $\kappa = c \lambda$ and $Z_0 = I_p$ and using geodesic convexity yields the next corollary.
\begin{cor}\label{cor:grad-bound}
Suppose $f$ is a $(1- \lambda)$-quantum expander. Then $\|\nabla f^{\Phi}(Z)\|_F \geq \lambda \min\{c \lambda, \|\log Z\|_F\}$.
\end{cor}

\subsection{Fast Sinkhorn from strong convexity}\label{subsec:sinkhorn-sc}
The goal of this section is to prove \cref{thm:expander-sinkhorn}. In \cref{subsec:sinkhorn-el}, we will use \cref{thm:expander-sinkhorn} in conjunction with \cref{thm:random-expander} to show that Tyler's iterative procedure converges linearly.

The theorem follows immediately from \cref{lem:level-set} and \cref{lem:sinkhorn-exponential}, which we will prove shortly. The next lemma is the classic, and completely standard, analysis of progress in the Sinkhorn scaling algorithm.

\begin{lem}\label{lem:sinkhorn-progress}Suppose $Z_0$, $Z$ are successive Sinkhorn scaling iterates. If $\|\nabla f^\Phi(Z_0)\|_F^2$ is at most $1$, then we have 
$$f^\Phi(Z) \leq f^\Phi(Z_0) - \frac{1}{6}\|\nabla f^\Phi(Z_0)\|_F^2.$$
\end{lem}

\begin{proof}

Recall \cref{dfn:sinkhorn-alg}. Set $X =\Phi(Z_{0})^{-1}$ so that $\frac{p}{n} Z \Phi^*(X) = I_p$.
Now 
\begin{align*}
f^\Phi(Z)&= \frac{p}{n}\log \det\Phi(Z) - \log \det (Z)\\
&= \frac{p}{n} \left(\log \det (\Phi(Z)X) - \log \det(X)\right) - \log \det (Z)\\
& \leq p  \log \left(\frac{1}{n} \tr \Phi(Z)X\right) - \frac{p}{n}\log \det(X) - \log \det (Z)\\
&= p  \log \left(\frac{1}{n} \tr Z\Phi^*(X)\right) - \frac{p}{n}\log \det(X) - \log \det (Z)\\
& = f^{\Phi}(Z_0) + \log \det \left(W\right),
\end{align*}
for $W:=\sqrt{Z_0} Z^{-1} \sqrt{Z_0}$. Here the inequality is Jensen's inequality and the last equality is by our choice of $Z_0$ and $Z$. Compute $$\nabla f^\Phi(Z_0) = \frac{p}{n}\sqrt{Z_0} \Phi^*(\Phi(Z_0)^{-1}) \sqrt{Z_0} -I_p = \frac{p}{n} \sqrt{Z_0}\Phi^*(X)\sqrt{Z_0} -I_p =   W -I_p.$$ 
Note that $\tr W = \tr Z_0 Z^{-1} = \frac{p}{n} \tr Z_0 \Phi^*(X)  = \frac{p}{n} \tr \Phi(Z_0) X = p$. By Lemma 5.1 of \cite{garg2017operator} (a robust version of the AM-GM inequality) 
\begin{gather*}\log \det(W) \leq - \frac{1}{6}\min\{1, \left\|W -I_p\right\|_F\} =
- \frac{1}{6}\|\nabla f^{\Phi}(Z_0)\|_F^2. \qedhere
\end{gather*}
\end{proof}

\begin{lem}\label{lem:sinkhorn-exponential} Suppose $\|\nabla f^{\Phi}(Z)\|_F \leq 1$. If $f^{\Phi}:\PD_1(p) \to \RR$ is geodesically $\lambda$-strongly convex in a sublevel set of $f$ containing $Z_0$, the $T^{th}$ Sinkhorn iterate $Z_T$ starting at $Z_0$ satisfies
$$\|\nabla f^{\Phi}(Z_T)\|_F^2 \leq 2^{- \lambda T/24}.$$
\end{lem}

\begin{proof} 
Define $f:=f^{\Phi}$. If $f$ is geodesically $\lambda$-strongly convex in a sublevel set $S \subset \PD_1(p)$ of $f$ containing $Z$ then by \cref{eq:zero-deriv},
$$f(Z^*) \geq f(Z_0) - \frac{1}{2\lambda} \|\nabla f(Z_0)\|_F^2.$$

By \cref{lem:sinkhorn-progress}, each step of Sinkhorn iteration decreases $f$. Because $S$ is a sublevel set, the normalized Sinkhorn iterates $Z_i$ remain in $S$. Recall that $f$ and its gradient are the same at the normalized and unnormalized Sinkhorn iterates. Moreover, if $\|\nabla f(Z_0)\|_F^2 = \eps> 0$, then by \cref{lem:sinkhorn-progress} the number of steps $T$ before $\|\nabla f(Z)\|_F^2 \leq \eps/2$ must satisfy 
$$ f(Z_0) - T \eps/12 \geq  f(Z^*) \geq f(Z_0) - \frac{1}{\lambda} \eps,$$
or $T \leq 12/\lambda$. Repeating this argument for any power of $2$ tells us that after $T$ steps $\|\nabla f(Z_T)\|_F^2 \leq 2^{- \lambda T/24}.$ 
\end{proof}

\begin{proof}[Proof of \cref{thm:expander-sinkhorn}] By \cref{lem:level-set} $f^{\Phi}$ is geodesically convex on a level set containing $Z_0$. By \cref{lem:sinkhorn-exponential}, the $T^{th}$ Sinkhorn iterate $Z$ starting at $Z_0$ satifies $ \|\nabla f^{\Phi}(Z)\|_F = \exp( - O(\lambda T))$. \end{proof}
\begin{rem}
\cref{lem:sinkhorn-exponential} applies for any descent method, such as geodesic gradient descent, satisfying the conclusion of \cref{lem:sinkhorn-progress}.
\end{rem}

\subsection{Fast Sinkhorn for elliptical distributions.}\label{subsec:sinkhorn-el}

Here we prove \cref{thm:fast-sinkhorn-elliptical}. First we show that the gradient quickly becomes less than a small constant on the iterates of Sinkhorn's algorithm. For short-hand, define 
$$ f_{\vec v}:=f^{\Phi_{\vec v}}.$$

\begin{lem}\label{lem:sinkhorn-initial} Suppose $\Phi$ is a $1- \lambda$-quantum expander. Then for some 
$$ T = O\left(\frac{1}{\eps^2}(f^\Phi(Z_0) - p \log (s(\Phi)/n))\right),$$
 the output $Z_T$ of Sinkhorn's algorithm starting at $Z_0$ satisfies 
$ \| \nabla f^{\Phi} (Z)\|_F \leq \eps.$
\end{lem}
\begin{proof}
By standard analyses of Sinkhorn's algorithm \cite{garg2017operator}, $T$ can be taken to be on the order of  
\begin{gather}\frac{1}{\eps^2} \left(f^\Phi(Z_0) - \inf_Z f^\Phi(Z)\right).\label{eq:potential}
\end{gather} 
Because $\Phi$ is doubly balanced, the infimum occurs at $Z = I_p$ and we have $\inf_Z f^\Phi(Z)= \frac{p}{n} \log \det  s(\Phi) I_n/n = p \log(s(\Phi)/n)$.
\end{proof}

Next, we show that a doubly balanced scaling of a sufficiently good approximate quantum expander is  a quantum expander. This is the expander upon which we will be applying \cref{thm:expander-sinkhorn}.
\begin{lem}\label{lem:scaled-expander}Suppose $\delta \leq c  \lambda^2/\log p$ and that $\Phi$ is a $(\delta, 1 -\lambda)$-quantum expander. Then $\Phi$ is scalable and for any $L, R$ such that $\Phi_{L,R}$ is doubly balanced, 
$\Phi_{L,R}$ is a $1 - \lambda + O(\delta \log p/\lambda)$ quantum expander.
\end{lem}
\begin{proof}
By \cref{thm:lapchi}, a unique minimizer $Y \in \PD_1(p)$ to $f^{\Phi}$ exists (i.e., a matrix $Y$ such that $Y^{1/2}$ and $\Phi(Y)^{-1/2}$ are scaling solutions) satisfying $\|I_p - Y\|_{op}, \|I_p - \Phi(Y)^{-1}\|_{op} = O( \delta \log p /\lambda)$. By \cref{lem:expansion-scaling}, the map $\Psi:=\Phi_{\sqrt{Y}, \Phi(Y)^{-1/2}}$ is a $(0, 1 - \lambda + O(\delta \log p/\lambda))$-quantum expander. By uniqueness of $Y$, $L^\dagger L, R^\dagger R$ match $Y, \Phi(Y)^{-1}$ up to scalar multiples, and it is easy to verify that expansion of $\Phi_{L,R}$ depends only on the equivalence class of $L^\dagger L, R^\dagger R$ under scalar multiplication.
\end{proof}

\begin{cor}\label{cor:scaled-expander} Let $\lambda$ be a small constant as in \cref{thm:random-expander}. With probability $1 - O(e^{-q(p,n,c\lambda^2/\log p)})$, 
\begin{gather}\left(\Phi_{\vec x}\right)_{\widehat{\Sigma}^{-1/2},\Phi_{\vec x}(\widehat{\Sigma}^{-1})^{-1/2}}\label{eq:scaled-vectors}
\end{gather} is a $1 - .5\lambda$-quantum expander. In other words, the operator $\Phi_{\vec z}$ where 
$$z_i = \widehat{\Sigma}^{-1/2} x_i/\|\widehat{\Sigma}^{-1/2} x_i\|$$ is a $1 - .5 \lambda$ quantum expander.
\end{cor}
\begin{proof}
Let $\lambda$ be a small constant as in \cref{thm:random-expander}. With probability at least $1 - O(e^{-q(p,n,c\lambda^2/\log p)}) = 1 - O(e^{ - \Omega(n/(p \log^2 p))})$ for $n \geq C p \log^2 p$, the operator $\Phi_{\vec v}$ is an $(c\lambda^2/\log p, 1-\lambda)$ quantum expander. As the operator in \cref{eq:scaled-vectors} is a doubly balanced scaling of $\Phi_{\vec v}$, \cref{lem:scaled-expander} implies the claim.
\end{proof}

We can now prove \cref{thm:fast-sinkhorn-elliptical}.

\begin{proof}
Let $\lambda$ be a small constant as in \cref{thm:random-expander}. By \cref{cor:scaled-expander}, with probability at least $1 - O(e^{-q(p,n,c\lambda^2/\log p)}) = 1 - O(e^{ - \Omega(n/(p \log^2 p))})$ for $n \geq C p \log^2 p$, the operator
$$\Psi:=\left(\Phi_{\vec x}\right)_{\widehat{\Sigma}^{-1/2},\Phi_{\vec x}(\widehat{\Sigma}^{-1})^{-1/2}}$$
is a $(1 -.5 \lambda)$-quantum expander. Moreover, as reasoned in the proof of \cref{thm:elliptical}, $\|\Sigma^{1/2} \widehat{\Sigma}^{-1} \Sigma^{1/2} - I_p\|_{op} = O( c\lambda^2)$. Because $\Psi$ is a suitable scaling of $\Phi_{\vec x}$, if the normalized Sinkhorn iterates of $\Phi_{\vec x}$ starting at $I_p$ are $Z_i$ then the normalized Sinkhorn iterates of $\Psi$ starting at $\widehat{\Sigma}$ are $Z'_i =\widehat{\Sigma}^{1/2} Z_i \widehat{\Sigma}^{1/2}$. By \cref{lem:sinkhorn-initial}, for some $T$ as in \cref{lem:sinkhorn-initial} we have $\|\nabla f^{\Psi}(Z'_{T})\|_F \leq \delta_0$ where $\delta_0 = c  \lambda^2$.

Because $\Psi$ is doubly balanced and $\tr \widehat{\Sigma} = p$,
\begin{align*}f^{\Psi}(\widehat{\Sigma}) - \log(s(\Phi)/n) &\leq p \log \left(\frac{1}{n}\tr \Psi(\widehat{\Sigma}) \right) - \log \det(\widehat{\Sigma}) - \log(s(\Phi)/n)\\
& = - \log \det (\widehat{\Sigma}),
\end{align*}
where the inequality is Jensen's. Next, using that $\widehat{\Sigma} \in \Sigma(1 + O(1))$ we have $|\log\det(\widehat{\Sigma})| = O(p) + |\log \det(\Sigma)|,$ so $T = O((p + |\log \det(\widehat{\Sigma})|/\delta_0^2) = O(p + |\log\det(\Sigma)|)$. 

By \cref{thm:expander-sinkhorn}, the number of further steps $T$ to obtain a normalized Sinkhorn iterate $Z'$ with $\|\nabla f^{\Psi}(Z')\|_F \leq \eps$ is $O(\log (1/\eps))$. We may assume $\eps \leq c \lambda$, because otherwise $\log (1/\eps) = \Theta(1)$. Next we apply \cref{cor:grad-bound} to $\Psi$ at $I_p$ to see that $\|\log (Z')\|_F = 2\eps/\lambda = O(\eps)$, so the corresponding normalized Sinkhorn iterate $Z$ satisfies 
 $$\|\log \widehat{\Sigma}_1^{1/2} Z \widehat{\Sigma}_1^{1/2}\|_F = O(\eps).$$ 
 Applying \cref{lem:trace-det} completes the proof. \end{proof}

\section{Quantum expansion for random unit vectors}\label{sec:random}

Here we prove \cref{thm:random-expander}. We prove the theorem in \cref{subsec:random-cheeger} after stating the main technical ingredient of the proof, \cref{lem:random-cheeger}.

\subsection{A Cheeger constant for operators.}

\begin{lem}\label{lem:unitaries} Let $\vec v = (v_1, \dots, v_n)  \in (\CC^p)^n$ be a tuple of vectors. Then $\Phi_{\vec v}$ is an $(\eps, \lambda)$-quantum expander if and only if $\Phi_{\vec v}$ is $\eps$-doubly balanced and for all unitaries $U$, the matrix $B^U \in \Mat(p, n)$ given by 
$$ B^U_{i,j} = B^U_{i,j} = |(Uv_i)_j|^2$$
satisfies 
$$\sup_{x \in \RR^n, \langle x, \vec 1_p \rangle = 0} \frac{\|B^Ux\|}{\|x\|} \leq (1 - \lambda) s(\Phi_v)/\sqrt{np}.$$
\end{lem}

\begin{proof}
Write
\begin{align*}\sup_{X \in\Herm(p), \tr X = 0} \frac{\|\Phi(X)\|_F^2}{\|X\|_F^2} &= \sup_{U \in U(p)} \sup_{x \in \RR^n: \langle x, \vec 1_p \rangle = 0}  \frac{\|\Phi(U^\dagger\diag(x) U)\|_F^2}{\|x\|^2}\\
&= \sup_{U \in U(p)} \sup_{x \in \RR^n: \langle x, \vec 1_p \rangle = 0} \frac{\sum_{i = 1}^n \langle Uv_i ,\diag(x) Uv_i \rangle^2 }{\|x\|^2}\\
&=\sup_{U \in U(p)} \sup_{x \in \RR^n: \langle x, \vec 1_p \rangle = 0} \frac{\sum_{i = 1}^n |(B^Ux)_i|^2 }{\|x\|^2}.\qedhere
\end{align*}
\end{proof}

\begin{dfn}
Let $B \in \Mat(p,n)$ be a nonnegative matrix. The Cheeger constant $\ch(B)$ of the weighted bipartite graph associated to $B$ is given by 
$$\ch(B):=\min_{T \subset [p], |T| \leq p/2, S \subset [n]} \phi(S,T)$$
where
$$\phi(S,T) := \frac{\cut(S, T)}{\min\{\vol(S,T), \vol(\overline{S}, \overline{T})\}}$$
where 
$$ \vol(S,T):= \sum_{i \in T, j \in [n]} b_{ij} + \sum_{i \in [p], j \in S} b_{ij} \textrm{ and } \cut(S, T):= \sum_{i \not\in T, j  \in S} b_{ij} + \sum_{i \in T, j \not\in S} b_{ij}.$$
\end{dfn}

\begin{lem}\label{lem:graph-cheeger}Let $B \in \Mat(p,n)$ be a nonnegative matrix. Suppose $B$ is $\eps$-doubly balanced. Then for $\eps < c \ch(B)^2$,
$$ \sup_{ x \in \RR^n, \langle x, \vec 1_p \rangle=0 } \frac{\|Bx\|_2}{\|x\|_2} \leq \max\left\{1/2, 1 - \ch(B)^2 + C \frac{\eps}{\ch(B)^2}\right\} \frac{s(B)}{\sqrt{np}}.$$ 
Here $C, c$ are absolute constants.
\end{lem}
\begin{proof}
By \cite{KLR19}, the spectral gap of $B$ is at most $\ch(B)^2 - 3 \eps.$ By \cref{lem:translator}, $B$ is a quantum expander with the desired parameters.
\end{proof}
\cref{lem:unitaries} and \cref{lem:graph-cheeger} lead us to consider the \emph{minimum} value 
$$\ch(\vec v):=\inf_{U} \ch(B^U),$$
and immediately imply the following:

\begin{cor} \label{cor:cheeger-expander}

Let $\eps < c \ch(\vec v)^2$. If the operator $\Phi_{\vec v}$ is $\eps$-doubly balanced, then $\Phi_{\vec v}$ is an
$$ \left(\eps, \max\left\{\frac12, 1 - \ch(\vec v)^2 + C \frac{\eps}{\ch(\vec v)^2}\right\}\right)$$
quantum expander. Here $C, c'$ are absolute constants.
\end{cor}

\begin{rem}[General operators] If we define $$\ch(\Phi):= \inf_{U \in U(p), V \in U(n)} \ch(B^{U,V})$$ where $B^{U,V}_{i,j} := e_j^\dagger \Phi_{U,V}(e_i e_i^\dagger)e_j$, then it is not hard to show \cref{cor:cheeger-expander} holds with $\Phi_{\vec v}$ replaced by $\Phi$ and $\ch(\vec v)$ replaced by $\ch(\Phi)$. Moreover, $\ch(\vec v) = \ch(\Phi_{\vec v})$.

\end{rem}

We next show how to express $\ch(\vec v)$ as an infimum over projections rather than unitaries. This will be easier to control in the random setting. 

\begin{dfn}[Conductance for projections]
Suppose $\vec v \in (\RR^p)^n$, $\pi:\RR^p \to \RR^p$ is an orthogonal projection, and $S \subset [n]$. We think of $S,\pi$ as a cut, and the quantity $\phi$ defined below as the conductance of the cut. Define 
$$ \phi(S, \pi) =  \frac{\cut(S, \pi)}{\min\{\vol(S, \pi), \vol(\overline{S}, I_p - \pi)\}.}$$
where 
$$
\vol(S, \pi) = \sum_{i \in n} \|\pi v_i \|^2 + \sum_{i \in S} \|v_i\|^2,
\textrm{ and } \cut(S, \pi) = \sum_{i \in S} \|(I_p - \pi) v_i \|^2 + \sum_{i \not\in S} \|\pi v_i \|^2 .$$

\end{dfn}

\begin{lem}\label{lem:operator-cheeger}
$\ch(\vec v) =  \min_{S, \pi:\rk \pi\leq p/2} \phi(S, \pi).$
\end{lem}
\begin{proof}
The conductance $\phi(S, T)$ the cut $S, T$ of $B^U$ is exactly the same as the conductance of $\phi(S, \pi)$ of $\vec v$ where $\pi$ is the orthogonal projection to the span $\{U e_i : i \in T\}$, which ranges over all orthogonal projections as $T, U$ range over $\binom{[p]}{\leq p/2}\times U(p)$. Hence, 
$$\min_U \min_{S, T: |T| \leq p/2} \phi(S, T) = \min_{S, \pi:\rk \pi \leq p/2} \phi(S, \pi).\qedhere$$
\end{proof}

\subsection{Probabilistic preliminaries}

Recall the Bernstein condition, which can be used to show from the moments of a random variable that it is subexponential. 
\begin{dfn}[Subexponential variables and Bernstein's condition]
Recall that random variable $X$ with mean $\mu$ is said to be $(\nu, b)$-\emph{subexponential} if 
$$ \EE[e^{\lambda( X - \mu)}] \leq e^{\frac{1}{2} \nu^2\lambda^2} \textrm{ for all } |\lambda| \leq 1/b .$$

Say a random variable $X$ with mean $\mu$ and variance $\sigma^2$ satisfies the \emph{Bernstein condition with parameter }$b$ if for all integers $k \geq 3$ we have
\begin{align}
|\EE[(X - \mu)^k]| \leq \frac{1}{2} k! \sigma^2 b^{k-2}.
\end{align}
It is known that a random variable satisfying the Bernstein condition with parameter $b$ is $(2\sigma, 2b)$-subexponential.
\end{dfn}
Recall that the $\operatorname{Beta}(\alpha, \beta)$ distribution has mean $\mu = \alpha/(\alpha + \beta)$, variance $\sigma^2 = \alpha \beta /(\alpha + \beta)^2(\alpha + \beta + 1)$, and $k^{th}$ raw moment 
$$ \prod_{r = 0}^{k - 1}\frac{\alpha + r}{\alpha + \beta + r}.$$ We will be concerned with the setting when $2\alpha \leq 2\beta$ are positive integers adding to $p$.

\begin{lem}\label{lem:bern} If $X$ is $\bbeta(1/2, p - 1/2)$, then $X$ is $(\nu,\beta)$-sub-exponential where $\nu, b =O(p^{-1})$.
\end{lem}

\begin{proof}
As $X$ has $\sigma^2 = O(p^{-2})$, it suffices to show that $X$ satisfies the Bernstein property with parameter $b = O(1/\beta)$.

Because $\mu$ and $X$ are nonnegative, $-\mu^k \leq (X - \mu)^k \leq X^k$. Thus, 
$$-\mu^k \leq \EE[(X - \mu)^k] \leq \EE[X^k].$$
Thus, $|\EE[(X - \mu)^k| \leq \max\{\mu^k, \EE[X^k]\} = \EE[X^k]$. The Bernstein property is satisfied if for $k \geq 3$ we have
\begin{align*}b^{k - 2} &\geq \frac{2}{k!} \EE[X^k]/\sigma^2\\
& =  \frac{2}{k!} \frac{(\alpha + \beta)^2 (\alpha + \beta + 1)}{\alpha\beta}\prod_{r = 0}^{k - 1}\frac{\alpha + r}{\alpha + \beta + r}\\
&= \frac{2}{k!} \frac{(\alpha + \beta) (\alpha + 1)}{\beta} \prod_{r = 2}^{k - 1}\frac{\alpha + r}{\alpha + \beta + r} \\
&= O\left( \frac{1}{k!} \prod_{r = 2}^{k - 1}\frac{\frac{1}{2} + r}{p + r}\right) \\
&= O\left(\frac{1}{k!} \frac{1}{p^{k-2}}\prod_{r = 2}^{k - 1}(r+1) \right)\\
&= O(1/p^{k-2}),
\end{align*}
so we may take $b = O(1/p)$. \end{proof}

\begin{lem}\label{lem:na} Let $v$ be a uniformly random element of $\Ess^{p-1}$, and define 
$$X_k = \sum_{i = 1}^k |v_i|^2 \textrm{ and } Y = \sum_{i = k+1}^l |v_i|^2$$ 
for $1 \leq k < l \leq p$. That is, $X_k$ and $Y$ are sums of squares of disjoint sets of coordinates of $v$. Then $X_k$ and $Y$ are negatively associated.

\end{lem}

\begin{proof} Recall that it suffices to show $\Pr[X \geq a \wedge Y \geq b] \leq \Pr[X \geq a] \Pr[Y \geq b]$ for all $a, b \in (0,1)$, or equivalently that $\Pr[Y \geq b|X \geq a] \leq \Pr[ Y \geq b]$ for all $a, b \in (0,1)$. 

We may sample $X$ and $Y$ by sampling independent $\chi$-squared random variables $Z_1, \dots, Z_p$ and setting $ X = X'/(X' + Y' + Z'), Y = Y'/(X' + Y' + Z')$ where $X' = \sum_{i = 1}^k Z_i, Y' = \sum_{i = k+1}^l Z_i,$ and $Z' = \sum_{i = l+1}^p Z_i$. In particular, $X', Y', Z'$ are independent. 

Now $X \geq a$ if and only if $X' \leq a(X' + Y' + Z')$, or $Y' \leq f(X', Z') := - Z' + (1 - a)X'/a $, and $Y' \geq b$ if and only if $Y' \leq b(X' + Y' + Z')$ or $Y' \geq g(X',Z'):=b(X' + Z')/(1 - b) \geq 0$. Thus, we seek to show that 
$$\Pr[Y' \geq g(X', Z')| Y' \leq f(X', Z')] \leq \Pr[ Y' \geq g(X', Z')].$$
To prove this, it suffices to show that $\Pr[Y' \geq \alpha | Y' \leq \beta] \leq \Pr[ Y' \geq \alpha]$ for all numbers $0 \leq \alpha < \beta$. This is true because the event $Y' \geq \alpha$ contains the complement of $Y' \leq \beta$.  \end{proof}

\begin{lem}\label{lem:subexp} Let $X_k$ be as in \cref{lem:na}, i.e. $X_k = \bbeta(k - 1/2, p - k - 1/2)$. Then $X_k$ is $(\nu = O(k^{1/2} p^{-1}), b = O(p^{-1}))$-subexponential.
\end{lem}
\begin{proof} By \cref{lem:bern}, it suffices to show $\EE[e^{\lambda (X_k - \mu_k)} ] \leq \EE[e^{\lambda (X_1 - \mu_k)} ]^k,$ for then 
$$ \EE[e^{\lambda (X_k - \mu_k)}] \leq e^{k \lambda^2 \nu^2/2}.$$
 This we show by induction. Clearly the claim holds for $k = 1$. For $k \geq 2,$ note that $X_k = X_{k - 1} + Y$ where $Y$ is as in the previous lemma. Note that the marginal distribution of $Y$ is that of $X_1$. 

Let $\EE[X_k] = \mu_k$. The function $e^{\lambda (X - \mu_k)}$ is either monotone increasing or decreasing in $X$; by \cref{lem:na}, $X_{k-1}$ and $Y$ are negatively associated and so 
$$\EE[ e^{\lambda (X_k - \mu_k}] =\EE[ e^{\lambda (X_{k -1} - \mu_{k-1})} e^{\lambda (Y - \mu_{1})}] \leq \EE[ e^{\lambda (X_{k -1} - \mu_{k-1})} ] \EE[e^{\lambda (X_{1} - \mu_{1})}].$$
This completes the proof of the inductive hypothesis.\end{proof}

\begin{thm}[Subexponential tail bound, \cite{wainwright2019high}]\label{thm:subexp}
If $Y_i$ are i.i.d $(\nu_i, b)$ subexponential, and $Y = \sum_{i = 1}^n Y_i$, then 
$$\Pr[|Y - \mu| \geq nt] \leq 
\left\{\begin{array}{cc} 
2e^{- n t^2/\nu_*^2} & \textrm{for }0 \leq t \leq b/\nu_*\\
2e^{- n t/2b} & \textrm{for } t > b/\nu_*\\
\end{array}\right.$$
for $\nu_* = \sqrt{\sum_{i = 1}^n \nu_i^2/n}.$

\end{thm}
By applying the previous theorem with $t =\eps  k/p $, $\nu_i = \nu= O(\sqrt{k}/p)$ and $b = O(1/p)$ we have the following corollary. 
\begin{cor}\label{cor:subexp} Let $Y_1, \dots, Y_n$ be independently distributed with each $Y_i \sim X_k$ for $X_k$ as in \cref{lem:subexp}. Then $Y = Y_1 + \dots + Y_n$ satisfies 
$$ \Pr[ (1 - \eps) nk/p \leq  Y \leq (1 + \eps) nk/p] \leq 2 e^{- \Omega( kn \eps ^2)} $$
for all $\eps \leq c$, where $c$ is a small enough constant.

\end{cor}

Finally, we use a standard result in random matrix theory. If $V$ is a Haar-random element of $\mathbb{S}^{p-1}$, then the distribution $\sqrt{p} V$ is in isotropic position and is $O(1)$ subgaussian \cite{marchal2017sub}, so the following theorem follows from Theorem 5.39 of  (as rewritten in Equation 5.25) of \cite{vershynin2010introduction}.
\begin{thm}[\cite{vershynin2010introduction}]\label{thm:vershinyn} Let $v_1,\dots, v_n$ independent, Haar random unit vectors. There are absolute constants $C, c>0$ such that for all $C\sqrt{p/n}\leq c\delta$, $\delta \leq 1$ we have
$$\left\|\frac{p}{n} \sum_{i = 1}^n v_i v_i^\dagger - I_p \right\|_{op} \leq \delta.$$
with probability at least 
$$1 - 2\exp\left( - (c\sqrt{n} \delta - C\sqrt{p})^2\right).$$
\end{thm}

\subsection{The Cheeger constant of random unit vectors}\label{subsec:random-cheeger}
The purpose of this section is to prove the following lemma, which in conjunction with \cref{cor:cheeger-expander} implies \cref{thm:random-expander}.

 \begin{lem}\label{lem:random-cheeger}For $\lambda$ at most some absolute constant, 
if $v_1, \dots, v_n$ are sampled i.i.d. from $\Ess^{p-1}$, then $\ch(\vec v) \geq \lambda$ with probability at least $1 - e^{ - \Omega(n) + O( p \log p)}$. 
\end{lem}

The proof, which is merely a more involved version of the proof of \cref{thm:vershinyn} in \cite{vershynin2010introduction}, combines concentration of the random variable $\phi(S, \pi)$ for fixed $S, \pi$ with a union bound over $2^{[n]} \times N$ where $N$ is an $\eps$-net for the projections. We proceed to prove concentration.

Before showing the detailed proof of the lemma, we use the statement to complete the proof of \cref{thm:random-expander}.
\begin{proof}[Proof of \cref{thm:random-expander}] By \cref{thm:vershinyn} applied to the distribution $\sqrt p V$ where $V$ is a Haar random unit vector, the operator $\Phi_{\vec v}$ is $\eps$-doubly balanced with probability $1 - 2e^{ - (c \sqrt{n} \eps - C\sqrt{p})^2}$. By \cref{lem:random-cheeger}, for some small enough absolute constant $\lambda_0>0$ the quantity $\ch(\vec v)$ is at least $\lambda_0$ with probability at least $1 - O(e^{ - cn + C p \log p})$. Both events occur with probability $1 - O(\exp( - \min \{(c \sqrt n \eps - C \sqrt{p})^2, cn - C p \log p\}))$. By \cref{cor:cheeger-expander}, there is an absolute constant $c_0$ such that if $\eps < c_0\lambda_0^2$ then $\Phi_{\vec v}$ is a $(\eps, 1- \lambda)$-quantum expander with probability $1 - O(\exp( - \min \{(c \sqrt n \eps - C \sqrt{p})^2, cn - C p \log p\}))$ for $\lambda > 0$ another absolute constant.
\end{proof}

\begin{lem}\label{lem:fixed-proj} Suppose $\pi$ is of rank $k \leq p/2$. Then with probability at least $1 - e^{-\Omega( kn)}$, 
\begin{align}
\min_{S \subset [n]} \Phi(S, \pi) &\geq c\label{eq:bot-small}\\
\textrm{and } .5 kn/p \leq \sum_{i =1}^n \|\pi' v_i\|^2 &\leq n -.5 kn/p\label{eq:bot-big} 
\end{align}
Here $c$ is some small enough constant.
\end{lem}

\begin{proof} We prove that the failure probability of \cref{eq:bot-big,eq:bot-small} are individually $O(e^{nkp})$; the claim will then follow by the union bound. We first bound the failure probability of \cref{eq:bot-big}. Note that $\sum_{i =1}^n \|\pi v_i\|^2$ is distributed as $W  = \sum_{i = 1}^n W_i$ for $W_i \sim X_k$. By \cref{cor:subexp}, we have $\Pr[ |W  - nk/p| \geq .5 nk/p] = O(e^{-nkp})$. Hence $\min\{W, n - W\} \geq .5 nk/p$ with probability at least $1-e^{-\Omega(nk)}$

We now bound the failure probabilty of \cref{eq:bot-small}. For fixed $S, \pi$, observe that $\phi$ is distributed as $$
\phi(S, \pi) \sim \frac{Y + Z}{\min\{\ell + W, 2n - \ell - W\}}$$
 where $Y = \sum_{i = 1}^\ell {Y_i}$ for $Y\sim X_{p - k}$ and $Z = \sum_{i = 1}^{n - \ell} Z_i$ for $Z_i \sim X_k$ with $Y$ and $Z$ independent, and $W  = \sum_{i = 1}^n W_i$ for $W_i \sim X_k$. 
 
Let $\alpha$ be a constant that we will make small compared to $c$.\\

\textbf{Case 1:} $\ell \geq \alpha n$.\\
Here we use the bound $\phi(S, \pi) \geq Y/n$. By \cref{cor:subexp}, with probability at least $1 -O(e^{- \Omega( \ell (p-k)}) = 1- O(e^{- \Omega(n p)})$, $Y$ is at least $.5\ell (p - k) /p \geq .25 \alpha n$. The number of $S$ is at most $2^n$, so with probability at least $ 1 - O(2^n e^{ - \Omega(p n)}) = 1 - O(e^{ - \Omega(pn)})$ there exists no $S$ with $\ell \geq \alpha n$ such that $\phi(S, \pi) \leq .25 \alpha$. \\

\textbf{Case 2:} $\ell \leq \alpha n.$\\
We claim two events $A$ and $B$ hold simultaneously with probability at least $1 - O(e^{-\Omega( k n)})$. Let $A$ be the event that $W + \ell \leq \ell + 1.5 nk/p = O(\max\{\ell, nk/p\})$. Let $B$ be the event that $ Y + Z$ is at least $\max\{.25 \ell, .5 nk/p\} = \Omega(\max\{\ell, nk/p\}).$ By \cref{cor:subexp} $A$ holds with probability at least $1 - O(e^{- \Omega(nk})$. We now bound the failure probability of $B$. If $.25 \ell \geq .5 nk/p$, then $Y$ is at least $.25 \ell $ with probability $1 - O(e^{-\Omega( l p)}) = 1 - O(e^{- \Omega( nk)})$, and $Z \geq .5 nk/p$ with probability $1 - O(e^{- \Omega( nk)})$. This shows that $B$ holds with probability $1 - O(e^{- \Omega(nk)})$, and by the union bound $A$ and $B$ hold simultaneously with probability $1 - O(e^{- \Omega(nk)})$. 

Condition on $A$ and $B$. For $\alpha \leq 1/8$, $A$ implies that $W \leq n - \ell$ and so $\min\{\ell + W, 2n - \ell - W\} = \ell + W$. Now $A$ and $B$ imply 
$$\phi(S, \pi) = \Omega\left( \frac{\max\{\ell, nk/p\}}{\max\{\ell, nk/p\}}\right) = \Omega(1).$$

Now, $A$ and $B$ fail for some $S$ with $\ell \leq \alpha n$ with probability on the order of
$$ e^{- ckn} \binom{n}{\leq \alpha n} \leq e^{ - \Omega(kn) + (\ln 2) H(\alpha) n} \leq e^{ - \Omega(kn)}$$
provided $H(\alpha)$, the binary entropy of $\alpha$, is a small enough constant. Hence, with probability $1 - e^{- \Omega(kn)}$, $\Phi(S, \pi) \geq c$ for all $S$ with $\ell \leq \alpha n$. \end{proof}

We now prove lemmas allowing us to construct and apply $\delta$-nets for the set of rank-$k$ orthogonal projections. Recall that a $\delta$-net for the projections of rank $k$ in the operator norm is a subset $N$ of projections such that for all projections $\pi:\RR^d \to \RR^d$ of rank $k$ there exists $\pi' \in N$ such that $\| \pi - \pi\|_{op} \leq \delta$.
\begin{lem}\label{lem:net-suffices} Suppose that a subset $N$ of the projections of rank $k$ is a $\delta$-net, and that for all $\pi' \in N$ we have
\begin{align}
 \min_{S \subset [n]} \phi(S, \pi') &\geq \beta.\label{eq:bot-small-1}\\
\textrm{and }.5 kn/p \leq \sum_{i =1}^n \|\pi' v_i\|^2 &\leq n -.5 kn/p \label{eq:bot-big-1} 
\end{align}
 Then for $\delta \leq \frac{c k }{p}$, we have $\min_{S \subset [n]}\phi(S, \pi) \geq \beta - O(p \delta/k)$ for all projections $\pi$ of rank $k$. Here $c$ is some small enough constant.
\end{lem}

\begin{proof}
Consider a rank $k$-projection $\pi$, and let $\pi \in N$ such that $\|\pi - \pi'\|_{op} \leq \delta$. Fix $S \subset [n]$.

Observe that 
$$\min\{\vol(S, \pi'), \vol(\overline{S}, I_p - \pi')\} \geq \min\{\sum_{i =1}^n \|\pi' v_i\|^2, n - \sum_{i =1}^n \|\pi' v_i\|^2\} \geq .5 kn/p.$$ Because $\phi(S, \pi')\geq \beta,$ we have $$\cut(S, \pi')  \geq .5\beta kn/p.$$
Because $\|\pi v_i\|^2 = \langle v_i, \pi v_i \rangle$ is linear in $\pi$, we have that $\|\pi'v_i\|^2 - \|\pi v_i\|^2 \leq \delta$ and similarly $\|(I_p - \pi')v_i\|^2 - \|(I_p - \pi)v_i\|^2\| \leq \delta$. We may then write
\begin{align*}
\cut(S, \pi) &\geq  \sum_{i \in S} \|(I_p - \pi') v_i \|^2 + \sum_{i \not\in S} \|\pi' v_i \|^2  - n\delta \\
&\geq \left(1 - \delta \frac{p}{.5 \beta k}\right)\cut(S, \pi').
\end{align*}
On the other hand, 
\begin{align*}
\vol(S, \pi) &= \sum_{i \in n} \|\pi v_i \|^2 + \ell \\
&\leq \sum_{i \in n} \|\pi' v_i \|^2 + \ell + n\delta\\
&\leq \vol(S, \pi') \left(1 + \delta \frac{p}{.5 k}\right)
\end{align*}
and similarly $
\vol(\overline{S}, I_p - \pi) \leq \vol(\overline{S}, I_p - \pi') \left(1 + \delta \frac{p}{.5 k}\right)$. Thus, 
$$ \phi(S, \pi) \geq \phi(S, \pi') \frac{1 - \delta \frac{p}{.5\beta k}}{1 + \delta \frac{p}{.5 k}} \geq \beta  - O(p \delta/k). \qedhere$$
\end{proof}

Following a standard method to prove the existence of $\delta$-nets, we consider a maximal code $N$ of operator norm distance $\delta/2$ in the set $X$ of rank $k$ projections. By the triangle inequality, there can be no point at distance at least $\delta$ from every point of $N$, else it could be added to the packing. We then use the Hamming bound to bound $|N|$.
\begin{lem}\label{lem:net-exists} There is an operator norm $\delta$-net $N$ of the rank $k$ orthogonal projections with 
$$|N| = \exp(O(pk |\ln \delta|)).$$
\end{lem}
\begin{proof} 
The rank $k$-projections in $\RR^{p}$ are of the form $X X^T$ where $X$ is a $d \times r$ matrix with $X^TX = 1$. 
In particular, $X$ has spectral norm $1$. Moreover, the image of a $\delta$ net $N$ for the $d \times r$ matrices of spectral norm $1$ under the map $X \to XX^T$ is a $2\delta$ net $N'$ for the rank $k$ positive-semidefinite matrices of spectral norm $1$ (which contains the set of rank $k$ projections). 
This is because for two such positive semidefinite matrices $XX^T$ and $YY^T$, $\|XX^T - YY^T\|_{op} \leq \|X(X^T - Y^T)\|_{op} + \|(X - Y)Y^T\|_{op} \leq 2 \|X - Y\|_{op}$. $N'$ can be used to obtain a $4\delta$ net $N''$ for the rank $k$ projections by taking one point from each nonempty intersection of a $2\delta$ ball about some element of $N'$ with the rank $k$ projections.  
A maximal code $N$ of distance $\delta/4$ for the spectral unit ball of $d\times r$ matrices is of order $O( (\delta/8)^{rd})$ by the Hamming bound. 
By the discussion preceding the statement of the theorem, $N$ is a $\delta/4$ net for the spectral unit ball of $d\times r$ matrices, and hence the aforementioned net $N''$ is a $\delta$ net for the rank $k$ projections.  
\end{proof}

Finally we prove the main result of the section.
\begin{proof}[Proof of \cref{lem:random-cheeger}] 
Fix $k \leq p/2$. By \cref{lem:net-exists}, there is a $c\lambda/p$ net $N$ with $|N| = \exp(O( pk \log p))$ for absolute constants $c, \lambda$. By \cref{lem:fixed-proj} and the union bound, the conditions \cref{eq:bot-small-1,eq:bot-big-1} in \cref{lem:net-suffices} hold with $\beta = 2\lambda $ for all $\pi \in N$ with probability $1 - e^{ - \Omega( nk)} e^{ O(pk \log p)}$. By \cref{lem:net-suffices}, $\min_{S \subset [n]} \phi(S, \pi) \geq \lambda$ for all projections $\pi$ of rank $k$ assuming $c$ is small enough compared to $\lambda$.

By the union bound over $k \in [p/2]$, $\ch(v) \geq \lambda$ with probability at least $1 - O( p e^{- \Omega( n)  + O(p \log p)}) = 1 - O(e^{- \Omega( n)  + O(p \log p)})$. 
\end{proof}

\section*{Acknowledgements}The first author would like to acknowledge Ankit Garg for an enlightening discussion, as well as Akshay Ramachandran for interesting conversations and for pointing out and helping fix an error in \cref{lem:net-exists}.

\appendix

\section{Spectral gap}\label{sec:spectral-gap}

\begin{dfn}[Spectral gap] Let $\Phi:\Mat(p) \to \Mat(n)$ be a completely positive map. Say $\Phi$ has spectral gap $\lambda$ if its second singular value is at most $(1 - \lambda) \frac{s(\Phi)}{\sqrt{np}}$.
\end{dfn}

Here is the theorem appearing in \cite{KLR19} from which we can show \cref{thm:lapchi}.
\begin{thm} \label{thm:actual-lapchi}
Let $p \leq n$. Suppose that the completey positive map $\Phi:\Mat(n) \to \Mat(p)$ is $\eps$-doubly balanced and has spectral gap $\lambda$ with $\eps \log p/\lambda^2$ is at most a small enough constant. Then the condition number of the scaling solutions $L$ and $R$ such that $\Phi_{L,R}$ is doubly stochastic satisfy 
$$ \kappa(L), \kappa(R) \leq 1 + O( \eps \log p /\lambda).$$
\end{thm}

In the next lemma, which is straightforward, we prove a relationship between the spectral gap and quantum expansion.
\begin{lem}\label{lem:translator} Let $\Phi$ be a completely positive map. There are constants $C,c>0$ such that the following holds for $\eps \leq c \lambda$.
\begin{itemize}
\item\label{it:exp-to-spec} If $\Phi$ is a $(\eps,1 - \lambda)$-quantum expander then it has spectral gap $\lambda - O(\eps)$, and 
\item\label{it:spec-to-exp} If $\Phi$ has spectral gap $\lambda$ and is $\eps$-doubly balanced then it is an $(\eps, \alpha)$-quantum expander for 
$$\alpha = \max\{1/2, 1 - \lambda\} + O(\eps/\lambda).$$
\end{itemize}
\end{lem}
Before proving \cref{lem:translator} let us use it to prove \cref{thm:lapchi}. 
\begin{proof}[Proof of \cref{thm:lapchi}] 
By \cref{it:exp-to-spec}, $\Phi$ has spectral gap $\lambda - O( \eps) = \Omega(\lambda)$. Apply \cref{thm:actual-lapchi} to $\Phi^*:\Mat(n) \to \Mat(p)$; because $\Phi^*$ is the adjoint of $\Phi$, the two have the same spectral gap. Hence the scaling factors $R$ and $L$ that make $(\Phi^*)_{R, L}$ doubly balanced satisfy $ \kappa(L), \kappa(R) \leq 1 + O( \eps \log p /\lambda).$ By scaling $R$ and $L$ by the respective geometric means of their singular values, we may assume that $\|I_p - L^\dagger L\|_{op}, \|I_p - R^\dagger R\|_{op} = O(\eps \log p /\lambda)$ and $\det L = \det R = 1$.
\end{proof}
\cref{lem:translator} is essentially an elementary statement about linear maps. Its proof uses the following technical linear algebra lemmas as well as one from \cite{KLR19}. 

\begin{lem}\label{lem:gap-to-quantum} Suppose $A \in \Mat(m,n)$ has spectral gap $\lambda$. Then there are constants $C,c > 0$ such that for every unit vector $x$ such that $\|A x\| > (1 - \delta)\|A\|_{op}$ with $\delta/\lambda < c$, we have 
$$ \|Ay\| \leq (\max\{1/2, 1 - \lambda\} + C\delta/\lambda)\|A\|_{op}$$
for all unit vectors $y \in x^\perp$.
\end{lem}

\begin{proof}
It is enough to prove the claim when $A$ has operator norm $1$. Let $v$ denote a top right singular vector of $A$. We assume $x$ is a unit vector. First we claim that $|\langle x, v \rangle| \geq \sqrt{1 - 2\delta/\lambda}.$
Write $x = \sqrt{(1 -\alpha)} uv + \sqrt{\alpha} w$ where $w \in v^\perp$ is a unit vector, $\alpha \in [0,1]$, and $u$ is in the complex unit circle. Because $Aw$ is in the orthogonal complement of $Av$, we have
\begin{align*}(1 - \delta)^2 \leq \|A x\|_2^2 = (1 -\alpha) \|Auv\|_2^2 + \alpha \|Aw\|_2^2 &\leq (1 -\alpha) 1 +\alpha(1 - \lambda)^2\\
&= 1 - \alpha(1 - (1 - \lambda)^2)
\end{align*}
Hence, $
\alpha \leq (1 - (1 - \delta)^2)/(1 - (1-\lambda)^2) \leq 2 \delta/\lambda.$ Next, for a unit vector $y \in x^\perp$ write $y = \sqrt{\beta} u' v + \sqrt{1 - \beta} w'$ where $w' \in v^\perp$ is a unit vector and $\beta \in [0,1]$. Because $\langle x, y \rangle = 0$, 
$$\sqrt{\beta}\sqrt{1 - \alpha} = \sqrt{1 - \beta} \sqrt{\alpha} |\langle w , w' \rangle| \leq \sqrt{1 - \beta} \sqrt{\alpha}.$$
Hence $\sqrt{\beta}/\sqrt{1-\beta} \leq \sqrt{\alpha}/\sqrt{1-\alpha}$ and so $\beta \leq \alpha \leq 2\delta/\lambda$. Take $\lambda' = \min\{\lambda, 1/2\}$. Now 
\begin{align*}\|A y\|_2^2 \leq \beta + (1 - \beta) (1- \lambda')^2 &\leq 2 \delta/\lambda + (1 - 2\delta/\lambda)(1 - \lambda')^2.
\end{align*}
Hence $\|A y\|_2 \leq (1 - \lambda')\sqrt{1 - 2\delta/\lambda + 2\delta/\lambda(1 - \lambda')^2} \leq (1 - \lambda')\sqrt{1 + 6\delta/\lambda} $, which completes the proof because we have assumed $\delta/\lambda$ is a small enough constant.\end{proof}

\begin{lem}\label{lem:quantum-to-gap}Suppose $A \in \Mat(m,n)$ is such that $\|A\|_{op} = 1$ and there exists a vector $x$ such that for all vectors $y \in x^\perp$, we have $\|Ay\| \leq (1 - \lambda) \|y\|.$ Then $A$ has spectral gap $\lambda$, i.e. $\sigma(A) \leq 1 - \lambda$.

\end{lem}
\begin{proof} We use the Rayleigh trace. Write 
$$1 + \sigma_2(A)^2 =  \sigma_1^2(A) + \sigma_2^2(A) = \sup_{\dim L = 2} \tr \pi_L A^\dagger A.$$
For any $L$, we may take $u_1, u_2$ to be an orthonormal basis spanning $L$ with $u_1 \in x^\perp$. Then 
$$\tr \pi_L A^\dagger A  = \|Au_1\|_F^2 + \|A u_2 \|_F^2 \leq (1 - \lambda)^2 + 1.$$ Thus, $\sigma_2(A) \leq 1 - \lambda$.
\end{proof}

\begin{lem}[Lemma 3.6 of \cite{KLR19}]\label{lem:sigma1}
Let $\Phi:\Mat(p)\to \Mat(n)$ be an $\eps$-doubly balanced completely positive map. The first singular value of the linear map $\Phi$ is at most $(1 + \eps) s(\Phi)/\sqrt{np}$. 
\end{lem}

\begin{proof}[Proof of \cref{lem:translator}]
Suppose $\Phi:\Mat(n) \to \Mat(m)$ is $\eps$-doubly balanced. By \cref{lem:sigma1}, the operator norm of $\Phi$ is at most $(1 + \eps) s(\Phi)/\sqrt{mn}$. Let $E$ denote the unit vector $I_n/\sqrt{n}$. 

First assume that $\Phi$ has spectral gap $\lambda$. Because $\Phi$ is $\eps$-doubly balanced, 
$$\|\Phi(E)\|_F \geq \|(1 - \eps)s(\Phi)I_m/m\sqrt{n}\|_F \geq (1 - \eps)s(\Phi)/\sqrt{mn}.$$
Hence, $\|\Phi(E)\|_F \geq ((1 - \eps)/(1 + \eps)) \|\Phi\|_{op}$. Applying \cref{lem:gap-to-quantum} to $\Phi$, $x = E$, and $\delta = 1 - (1 - \eps)/(1 + \eps) \leq 2 \eps$, we have that $\|\Phi(Y)\| \leq (\max\{1/2, 1 - \lambda\} + C\eps/\lambda)\|Y\|_F$ for all $Y$ orthogonal to $E$ and hence $\Phi$ is an $(\eps, \max\{1/2, 1 - \lambda\} + C\eps/\lambda)$-quantum expander. \\

On the other hand, if $\Phi$ is a $(\eps, 1 - \lambda)$-quantum expander, then $\|\Phi(Y)\|_F \leq \|Y\|_F(1 - \lambda)s(\Phi)/\sqrt{mn}$ for all $Y$ orthogonal to $E$, and hence $\|\Phi(Y)\|_F \leq \|Y\|_F(1 - \lambda)\|\Phi\|_{op}/(1 - \eps)$ for all $Y$ orthogonal to $E$. By \cref{lem:quantum-to-gap} applied with $\Phi$ normalized to have operator norm $1$ we have that $\Phi$ has spectral gap $1 - (1 -\lambda)/(1 - \eps) \geq \lambda - C \eps.$
\end{proof}

\section{Finite precision}\label{sec:finite-precision}

Here we show that even with access to only finitely many bits of $v_1, \dots, v_n$, Tyler's M-estimator $\widehat{\Sigma}$ still exists with high probability and remains close to the true shape $\Sigma$. Let $\vec v'$ denote the result of rounding the entries of $\vec v$.

Recall the proof of \cref{thm:elliptical}. To show that $\widehat{\Sigma}$ is close to $\Sigma$, it sufficed to show that Tyler's M-estimator on $\vec v$ is close to the identity. Similarly, to show $\widehat{\Sigma}$ is close to $\Sigma$ even after rounding, it suffices to show that Tyler's M-estimator on $\vec v'$ is close to the identity. This would follow from \cref{thm:random-expander} with $\vec v$ replaced by $\vec v'$, or more precisely that if $\vec v$ is chosen at random then $\phi_{\vec v'}$ is a $(\eps, 1- \lambda)$ quantum expander with high probability. As \cref{thm:random-expander} already shows $\phi_{\vec v}$ is an $(\eps, 1 - \lambda)$ quantum expander with the desired probability, the following lemma is enough.

\begin{lem}\label{lem:round-expander}
Suppose $p \leq n$, that $\vec v, \vec v'$ are such that $\|v_i - v'_i\| \leq \delta$ for all $i \in [n]$ and that $\Phi_{\vec v}$ is an $(\eps, 1 - \lambda)$ quantum expander. Then $\Phi_{\vec v'}$ is an 
$$(\eps - \delta n , 1 - \lambda' + \delta \sqrt{pn})$$ quantum expander.
\end{lem}
\begin{proof}
If for some $\alpha \leq \max\{c \lambda, 1/p\}$ we can show 
\begin{equation}\|\Phi_{\vec v} - \Phi_{\vec v'}\|_{op}\leq \alpha s(\Phi_{\vec v}), \label{eq:operator}\end{equation}
then $s(\Phi_{\vec v})(1 - \alpha \sqrt{p}) \leq s(\Phi_{\vec v'}) \leq (1 + \alpha \sqrt{p})s(\Phi_{\vec v})$, and so 

$\|\Phi_{\vec v'}(I_p) -\Phi_{\vec v}(I_p)\|_{op} \leq \alpha \sqrt{p}\; s(\Phi)$, $\|\Phi_{\vec v'}^*(I_n) -\Phi_{\vec v}^*(I_n)\|_{op} \leq \alpha \sqrt{n}\; s(\Phi_{\vec v'})$, and finally 
\begin{align*}\sup_{\tr X = 0, X \in \Herm(p)}\frac{ \|\Phi_{\vec v'}(X)\|_F}{\|X\|_F} &\leq (\lambda + \alpha) s(\Phi_{\vec v})\\
& \leq \frac{ (\lambda + \alpha)}{(1 - \alpha \sqrt{p})} s(\Phi_{\vec v}).
\end{align*}
Plugging in these parameters and using $p \leq n$ tells us $\Phi_{\vec v'}$ is an 
$$(\eps + O(\alpha \sqrt{n}), \lambda - O(\alpha \sqrt{p}))$$ quantum expander. We now upper bound $\alpha$ in \cref{eq:operator}. We may write $\Phi_{\vec v'} = \Phi_{\vec v} + T + S + \Phi_{\vec v' - \vec v}$ where the operators $S,T$ are given by $S(X)_{i,i} = v_i^\dagger X (v_i' - v)$ and $T(X)_{i,i} = (v'_i - v)^\dagger X v_i$ and $\Phi_{\vec v' - \vec v}$. Because $\|S\|_{op}= \|T\|_{op} \leq \delta \sqrt{n} $ and $\|\Phi_{\vec v - \vec v'}\|_{op} \leq \delta^2 \sqrt{n}$, we may take $\alpha = \delta \sqrt{n}$. \end{proof}
Now let $\vec{x}'$ denote the result of rounding $\vec{x}$ to $b$ bits after the decimal place so that each entry of $\vec x' -\vec x$ is at most $2^{-b}$. From the discussion at the beginning of this section, we have the following results.

\begin{thm}\label{thm:finite-precision}
Let $M\geq 1$ be a bound on the condition number of $\Sigma$, let $\vec x = x_1, \dots, x_n$ be drawn from an elliptical distribution with shape $\Sigma$, and let $\vec x' = x'_1, \dot, x'_n$ be the result of rounding $x_i$ to 
$$b_i \geq C \log \left(\frac{M np \log p}{\eps\|x_i\|}\right)$$ bits after the decimal place. Let $c > 0$ be an absolute constant.
\begin{enumerate}
\item\label{it:elliptical-rounding} With probability $1  - O(e^{-q(p,n,\eps/\log p)})$ for $q(p,n,\eps)$ is as in \cref{thm:random-expander}, the estimator $\tilde{\Sigma} = \widehat{\Sigma}(\vec x')$ on the rounded vectors satisfies 
\begin{gather}\|I_p - \Sigma^{1/2} \tilde{\Sigma}^{-1} \Sigma^{1/2}\|_{op} = O(\eps)\label{eq:round-bound}\end{gather}
provided $\delta \leq c/\sqrt{p}$.
\item\label{it:sinkhorn-rounding} With probability $1 - O(e^{-q(p,n,c/\log p)})$, Sinkhorn scaling on $\vec x'$ outputs $\overline{\Sigma}$ satisfying
\begin{gather} \|I_p - \widehat{\Sigma}^{1/2} \overline{\Sigma}^{-1} \widehat{\Sigma}^{1/2} \|_{op} = O(\eps) \label{eq:sink-bound}\end{gather}
 in  $O(|\log \det \Sigma| + p +  \log (1/\eps))$
 iterations.\end{enumerate}
\end{thm}

\begin{proof}
Let $\lambda$ be a small constant as in \cref{thm:random-expander}. By \cref{thm:random-expander}, with probability at least $1  - O(e^{-q(p,n,\eps/\log p)})$, the operator $\Phi_{\vec v}$ is an $(\eps/\log p, 1 - \lambda)$-quantum expander. Condition on this event. The proof of \cref{thm:elliptical} shows that $\|I_p - \Sigma^{1/2} \widehat{\Sigma}^{-1} \Sigma^{1/2} \|_{op} \leq \eps$. Let $\tilde{\Sigma}$ denote Tyler's M-estimator for the rounded vectors $\vec x'$. To prove \cref{it:elliptical-rounding}, by \cref{lem:triangle-ineq} it suffices to show that $\tilde{\Sigma}$ is close to $\widehat{\Sigma}$. First note that $\widehat{\Sigma}$ has condition number $O(M)$ because $\|I_p - \Sigma^{1/2} \widehat{\Sigma}^{-1} \Sigma^{1/2} \|_{op} \leq 1 $.

By \cref{cor:scaled-expander}, 
$$\Psi:=(\Phi_{\vec x})_{\widehat{\Sigma}^{-1/2}, \Phi_{\vec x}(\widehat{\Sigma}^{-1})^{-1/2}}$$ is a $1 - .5 \lambda$ quantum expander. Note that $\Psi = \Phi_{\vec z}$ for $z_i = \sqrt{Z} x_i/\|\sqrt{Z}x_i\|$, $Z:=\widehat{\Sigma}^{-1}$. Now if we let $z'_i = \sqrt{Z} x'_i/\|\sqrt{Z}x_i'\|$, observe that $z_i,z'_i$ are scale-invariant in $x_i, x'_i$, and set $x = x_i/\|x_i\| = 1, x' = x_i'/\|x'_i\|$ so that $\|x - x'\| \leq 2^{-b}/\|x_i\|$, then $\|z_i - z'_i\|$ is  
\begin{align*}\left\|\frac{\sqrt{Z} x x^\dagger \sqrt{Z}}{\langle x, Zx\rangle} - \frac{\sqrt{Z} x' (x')^\dagger \sqrt{Z}}{\langle x', Zx'\rangle} \right\|_F\\
\leq \|Z\|_{op} \frac{\| \langle  x', Z x' \rangle x x^\dagger -  \langle  x, Z x \rangle x' (x')^\dagger\|_F}{\langle x, Zx \rangle \langle x',Z x' \rangle.}\\
\leq \|Z\|_{op}\|Z^{-1}\|^2_{op} \| \langle  x', Z x' \rangle x x^\dagger -  \langle  x, Z x \rangle x' (x')^\dagger\|_F\\
\leq \|Z\|_{op}\|Z^{-1}\|^2_{op} (\langle  x', Z x' \rangle \| x x^\dagger - x' (x')^\dagger\|_F + \|x\|^2|\tr Z (xx^\dagger - x(x')^\dagger)|)  \\
\leq (CM)^2 \cdot O(\| x x^\dagger - x' (x')^\dagger\|_F )
= O( (CM)^2 \sqrt{p}2^{-b}/\|x_i\|),
\end{align*}

By our choice of $b_i$ and \cref{lem:round-expander}, the operator
$$\Psi':=(\Phi_{\vec x'})_{\widehat{\Sigma}^{-1/2}, \Phi_{\vec x'}(\widehat{\Sigma}^{-1})^{-1/2}}$$
is a $(\eps/\log p, 1 - .25 \lambda)$ quantum expander. By \cref{thm:lapchi}, the left scaling $Y^{1/2} \in \PD_1(p)$ for $\Psi'$ satisfies $\|I_p - Y\|_{op} = O(\eps)$. Then $Y^{1/2} \widehat{\Sigma}^{-1/2}$ is a left scaling for $\Phi_{\vec x'}$, i.e. $\tilde{\Sigma}^{-1}_1 = \widehat{\Sigma}^{-1/2}_1 Y \widehat{\Sigma}^{-1/2}_1$. We conclude that 
\begin{gather}\|\widehat{\Sigma}_1^{-1/2}\tilde{\Sigma}_1\widehat{\Sigma}^{-1/2}_1 - I_p\|_{op} = O(\eps).\label{eq:rounded}\end{gather}
\cref{eq:round-bound} follows from \cref{eq:rounded} and \cref{lem:trace-det}.

To prove \cref{it:sinkhorn-rounding}, condition on the operator $\Phi_{\vec v}$ being an $(c/\log p, 1 - \lambda)$-quantum expander for $c$ a small enough constant. This happens with probability $1  - O(e^{-q(p,n,c/\log p)})$. Let $Z_T$ denote the result of running normalized Sinkhorn iteration on $\Phi_{\vec x'}$. Let $T$ be large enough that $\|I_p - \tilde{\Sigma}_1^{1/2} Z_T \tilde{\Sigma}_1^{1/2} \|_{op} \leq \eps$; by the proof of \cref{thm:fast-sinkhorn-elliptical}, $T$ satisfies \cref{eq:sinkhorn-bound} in \cref{thm:fast-sinkhorn-elliptical}. \cref{eq:sink-bound} follows from \cref{eq:rounded}, \cref{lem:trace-det}, and \cref{lem:triangle-ineq}. \end{proof}

\section{Normalization conventions and error metric}
Define $d(A, B) := \|I_p - A^{1/2} B^{-1} A^{1/2}\|_{op}$ for $A, B \in \PD(p)$. Note that 
$$d(A, B) \leq \eps \iff  B^{-1} \in A^{-1}( 1 +  \eps),$$
where for $X \in\Herm(p)$ the short-hand $X(1 \pm \eps)$ denotes the interval of $Y \in \Herm(p)$ such that 
$X(1 - \eps)\preceq Y  \preceq X( 1 +  \eps).$

\begin{lem}\label{lem:triangle-ineq}
Let $A, B, C \in \PD(p)$.
\begin{itemize}
\item provided $d(A, B), d(B, C)$ are small enough constants, 
$$d(A, C) = O(d(A, B) + d(B, C)).$$
\item provided $d(A, B)$ is at most a small enough constant, 
$$d(B,A) = O(d(A, B)).$$
\item provided $d(A, B)$ is at most a small enough constant, 
$$d(A^{-1},B^{-1}) = O(d(A, B)).$$

\end{itemize}
\end{lem}
\begin{proof} The first item follows because 
\begin{align*}C^{-1} &\in B^{-1}(1 \pm d(B, C))\\
& \subset A^{-1}(1 \pm d(A, B))(1 \pm d(B, C))\\
& \subset A^{-1}(1 \pm O(d(A, B) + d(B, C)).
\end{align*}
The second item follows because $B^{-1} \in A^{-1}(1 \pm d(A, B))$ implies $A^{-1} \in B^{-1}(1 \pm O(d(A, B)))$ provided $\eps$ is small enough. The third item follows because for $A, B \in \PD(p)$ we have $A \preceq B$ if and only if $A^{-1} \succeq B^{-1}$.\end{proof}

\begin{lem}\label{lem:trace-det}
Suppose $A, B \in \PD(p)$ satisfy $\tr A = \tr B = 1$, and let $A_1 = \det(A)^{-1/p} A, B_1 = \det(B)^{-1/p} B \in \PD_1(p)$. Then if $d(A_1, B_1)$ is at most a small constant,
$$d(A, B) = O(d(A_1, B_1)).$$
\end{lem}
\begin{proof}
Note that $A = p A_1/\tr A_1, B = p B_1/\tr B_1$. Because for $A, B \in \PD(p)$ we have $A \preceq B$ if and only if $A^{-1} \succeq B^{-1}$, we have that $A \in B(1 \pm O(d(A_1, B_1))$. This tells us $\tr A \in (1 \pm O(d(A_1, B_1)) \tr B$, so 
$$A^{-1/2} B A^{-1/2} \in (1 \pm O(d(A_1, B_1)) A_1^{-1/2} B_1 A_1^{-1/2} \subset (1 \pm O(d(A_1, B_1))) I_p.$$ This completes the proof. 
\end{proof}

\section{Existence and uniqueness}\label{sec:king}

Here we prove \cref{it:tyler-uniqueness} of \cref{thm:king}. 
\begin{proof}
For the sufficiency of the condition we refer to \cite{Tyler_1987}. It remains to prove that it is necessary. Suppose \cref{eq:tyler} has a solution $A \in \PD(p)$ and let $y_i = A^{-1/2}x_i/\|A^{-1/2} x_i\|$ so that $y_i$ are in radial isotropic position, i.e. $\|y_i\| = 1$ and $\frac{p}{n} \sum_{i = 1}^n y_i y_i^\dagger = I_p$. Suppose the condition in \cref{it:tyler-uniqueness} is violated. Clearly, because $A$ is a linear transformation, the condition is violated for $y_i$ also, i.e. there is a $k$-dimensional subspace $0 \subsetneq W \subsetneq \RR^p$ containing at least $nk/p$ points $y_i$. Let $\pi_{W}$ denote the orthogonal projection to $W$. Then 
$$\tr \pi_{W} = \frac{p}{n} \sum_{i = 1}^n \tr \pi_{W} y_i y_i^\dagger \geq \frac{p}{n} (n k/p), $$
but both sides are equal to $k$ and so the inequality must be an equality. Equality can hold if and only if $\pi_W y_i = 0$ for all $y_i \not\in W$. That is, for every $i$ we have that $y_i \in W$ or $y_i \in W^\perp$. Then for any $a > 0$, the matrix $B = (a \pi_W + a^{-1} \pi_{W^\perp}) $ satisfies $B y_i/\|By_i\| = y_i$. Thus $A^{1/2} B^{-1} A^{1/2}$, renormalized if necessary, also solves \cref{eq:tyler}. By the invertibility of $A$, $A^{1/2} B^{-1} A^{1/2}$ constitutes a family of distinct solutions to \cref{eq:tyler}, none of which is a constant multiple of the other.
\end{proof}

\bibliographystyle{alphaurl}
\bibliography{BL-tyler-M}

\end{document}